\pdfoutput=1
\documentclass[11pt,a4paper]{article}
\usepackage[a4paper,margin=2.5cm]{geometry}
\usepackage{authblk}
\usepackage[labelfont=bf,labelsep=period]{caption}

\usepackage{amsmath,amssymb,amsthm,amsfonts}
\usepackage{txfonts}
\usepackage{float}
\usepackage{booktabs}
\usepackage{array}
\usepackage[section]{placeins}
\usepackage{algorithm}
\usepackage{algpseudocode}
\usepackage{graphicx}
\usepackage{bm}
\usepackage{cite}
\usepackage[colorlinks,linkcolor=black,citecolor=black,urlcolor=black,breaklinks]{hyperref}

\usepackage{xcolor}
\usepackage{pgfplots}
\pgfplotsset{compat=1.18}
\usepgfplotslibrary{groupplots}
\usetikzlibrary{calc}

\definecolor{figBlue}{HTML}{2A78D6}
\definecolor{figOrange}{HTML}{EB6834}
\definecolor{figGreen}{HTML}{1BAF7A}
\definecolor{figYellow}{HTML}{EDA100}
\definecolor{figPink}{HTML}{E87BA4}
\definecolor{figInk}{HTML}{0B0B0B}
\definecolor{figInkTwo}{HTML}{52514E}
\definecolor{figMuted}{HTML}{898781}
\definecolor{figGrid}{HTML}{E1E0D9}
\definecolor{figBase}{HTML}{C3C2B7}
\definecolor{figWash}{HTML}{F0EFEC}
\definecolor{rampA}{HTML}{86B6EF}
\definecolor{rampB}{HTML}{3987E5}
\definecolor{rampC}{HTML}{1C5CAB}
\definecolor{rampD}{HTML}{0D366B}
\pgfplotsset{
  paperfig/.style={scale only axis, clip=false,
    axis line style={draw=none}, tick style={draw=none},
    grid=major, major grid style={draw=figGrid, line width=0.3pt},
    tick label style={font=\scriptsize, text=figInkTwo},
    label style={font=\footnotesize, text=figInkTwo},
    title style={font=\footnotesize\bfseries, text=figInk, anchor=south west, at={(0,1)}, inner xsep=0pt}},
  figline/.style={line width=0.9pt, line join=round, line cap=round},
  figdot/.style={mark=*, mark size=1.8pt, mark options={draw=white, line width=0.3pt}}}
\tikzset{figlabel/.style={font=\scriptsize, text=figInkTwo, inner sep=1pt}}

\emergencystretch=\maxdimen
\allowdisplaybreaks

\numberwithin{equation}{section}
\newtheorem{theorem}{Theorem}[section]
\newtheorem{lemma}[theorem]{Lemma}
\newtheorem{proposition}[theorem]{Proposition}

\theoremstyle{definition}
\newtheorem{definition}[theorem]{Definition}
\newtheorem{remark}[theorem]{Remark}
\newtheorem{example}[theorem]{Example}

\newcommand{\R}{\mathbb{R}}
\newcommand{\norm}[1]{\lVert #1\rVert}
\newcommand{\ip}[2]{\langle #1,#2\rangle}
\newcommand{\minShareFD}{0.058}
\newcommand{\minShareFDsector}{finance}
\newcommand{\minShareFDyear}{2014}
\newcommand{\nSectors}{19}
\newcommand{\outputTrn}{173.3}
\newcommand{\finalTrn}{113.0}
\newcommand{\gdpObsTrn}{75.70}

\newcommand{\miningShareOutput}{27.6}
\newcommand{\miningExportShare}{96}
\newcommand{\specRadius}{0.389}
\newcommand{\maxBalErr}{0.4}
\newcommand{\specRadiusMax}{0.466}
\newcommand{\colSumMax}{0.639}
\newcommand{\rhoMax}{1.393}
\newcommand{\rhoMaxSector}{Education}
\newcommand{\rhoMaxShare}{97}
\newcommand{\rhoSecond}{Arts}
\newcommand{\rhoThird}{Health}
\newcommand{\rhoFourth}{Mining}
\newcommand{\rhoMin}{0.742}
\newcommand{\rhoMinSector}{Electricity}
\newcommand{\selfUseMin}{31}
\newcommand{\rasmussenKey}{electricity, transport, professional and administrative}
\newcommand{\fstarTrn}{80.81}
\newcommand{\effRatio}{0.937}
\newcommand{\gdpGap}{6.8}
\newcommand{\dualImports}{0.821}
\newcommand{\dualWages}{1.129}
\newcommand{\nYup}{11}
\newcommand{\nYlow}{6}

\newcommand{\xNearSector}{public administration}
\newcommand{\xNearPct}{9.997}
\newcommand{\kappaPct}{10}
\newcommand{\thetaPct}{5}
\newcommand{\rAtLP}{0.0000}

\newcommand{\epsObs}{6.32}
\newcommand{\centreFreeList}{transport, public administration and education}
\newcommand{\centreFreeMax}{2.19}
\newcommand{\etaBase}{6.22}
\newcommand{\epsPct}{0.5}
\newcommand{\deltaPct}{1}
\newcommand{\rBase}{0.673}
\newcommand{\boxBase}{0.109}
\newcommand{\rCube}{0.33}
\newcommand{\NDim}{38}
\newcommand{\gdpCentre}{99.73}
\newcommand{\distCentreLP}{0.039}
\newcommand{\distCentreObs}{0.495}
\newcommand{\rAbsBn}{5.0}
\newcommand{\absMinPct}{0.002}
\newcommand{\absMaxPct}{0.4}
\newcommand{\absMinSector}{Mining}
\newcommand{\absMaxSector}{Water supply}
\newcommand{\yUpList}{agriculture, mining, manufacturing, trade, accommodation, ICT, finance, real estate, public administration, arts and other services}
\newcommand{\yLowList}{electricity, water supply, construction, professional, administrative and health}
\newcommand{\adjTopList}{electricity, public administration, arts, mining and health}
\newcommand{\adjMin}{0.55}
\newcommand{\adjMax}{0.78}
\newcommand{\epsObsPlateau}{6.57}
\newcommand{\rPlateau}{0.718}
\newcommand{\epsKinkBase}{0.56}
\newcommand{\kinkRatio}{0.57}
\newcommand{\rEpsZero}{0.235}
\newcommand{\slopeEps}{0.87}
\newcommand{\dualGDP}{0.248}
\newcommand{\dualTopTwoName}{surplus side of the mining balance}
\newcommand{\dualTopTwo}{0.148}
\newcommand{\effMin}{0.881}
\newcommand{\effMax}{0.971}
\newcommand{\rSensMin}{0.386}
\newcommand{\rSensMax}{0.696}
\newcommand{\rThetaFive}{0.673}
\newcommand{\rThetaTen}{0.404}
\newcommand{\activeFive}{19}
\newcommand{\activeTen}{40}

\newcommand{\xBoundsTen}{19}

\newcommand{\fstarThetaTenChg}{3.0}
\newcommand{\rOneSidedBase}{0.359}

\newcommand{\rOneSidedTheory}{0.359}
\newcommand{\mcNsim}{20,000}
\newcommand{\mcFullSphere}{0.30}

\newcommand{\mcHalfSphere}{0.64}
\newcommand{\mcSphereFour}{99.8}
\newcommand{\mcSphereFive}{96.5}
\newcommand{\mcCubeFour}{16.9}
\newcommand{\mcCubeFive}{1.9}
\newcommand{\effPanelMin}{0.932}
\newcommand{\effPanelMax}{0.937}
\newcommand{\rPanelMin}{0.661}
\newcommand{\rPanelMax}{0.701}
\newcommand{\miningExpFirst}{69}
\newcommand{\miningExpLast}{96}
\newcommand{\rhoMiningFirst}{1.326}
\newcommand{\rhoMiningLast}{1.344}
\newcommand{\toyR}{3.43}
\newcommand{\toyCx}{104.28}
\newcommand{\toyCy}{62.57}
\newcommand{\toyCxMin}{103.43}
\newcommand{\toyCxMax}{104.28}
\newcommand{\rVersion}{4.6.1}
\newcommand{\highsVersion}{1.14.0.2}
\newcommand{\quadprogVersion}{1.5.8}

\begin{document}
\setcounter{page}{1}
\title{Input--output analysis with a sphere packing approach}

\author[1]{Munkhbayar Byambadash}
\author[2,*]{Tumendelger Lkhagvasuren}
\author[3]{Enkhbat Rentsen}
\affil[1]{The University of Finance and Economics, Ulaanbaatar, Mongolia}
\affil[2]{The Graduate School, University of Finance and Economics, Ulaanbaatar, Mongolia}
\affil[3]{Institute of Mathematics and Digital Technology, Mongolian Academy of Sciences, Ulaanbaatar, Mongolia}
\affil[*]{Corresponding author. Email: \href{mailto:tumendelger.l@ufe.edu.mn}{tumendelger.l@ufe.edu.mn}}
\date{}

\maketitle

\begin{abstract}
Input--output analysis with linear programming is widely used in the planning of national economies. However, linear programming attains its optimum at a vertex of the feasible set, and a small deviation of the optimal plan in some direction violates a constraint. So far, little attention has been paid to finding sets of production plans which guarantee both the supply--demand balance of an economy and a given level of its gross domestic product (GDP), and the sphere packing approach cannot be applied directly, since the balance equations define a set with an empty interior. In this paper, we introduce the set of $(\varepsilon,\delta)$-solutions, in which GDP may be smaller than its optimal value by at most $\varepsilon$ and supply may differ from demand by at most $\delta$, and we find a sphere with the maximum radius inscribed into this set by solving a linear programming problem. We prove that the optimal radius is positive for positive tolerances, that it is non--decreasing, concave and piecewise linear with respect to the tolerances, and that it is bounded by the narrowest balance constraint, which is determined by a sales-concentration index of sectors. The approach is applied to the official input--output tables of Mongolia for 2010--2019 and 2024. The classical optimal plan has zero robustness, while the robust plan keeps \gdpCentre\% of the maximal GDP and, in a Monte Carlo simulation, tolerates independent deviations of all outputs and final demands of up to $\pm$\mcFullSphere\%. The margin of robust plans is limited by sectors with concentrated sales, mainly education and, increasingly, export-oriented mining, which indicates that a diversification of the users of mining output widens the set of robust plans.


\end{abstract}

\section{Introduction}

\subsection{Background and motivation}\label{sec:background}

Input--output (IO) analysis introduced by Leontief \cite{Leontief1936} describes an economy as a system of interdependent sectors, where each sector buys intermediate inputs from other sectors and sells its output to them and to final users \cite{MillerBlair2022}. Since every sector is both a supplier and a user, a production plan of an economy must satisfy a system of balance equations: the supply of every sector must be equal to the demand for its output. Linear programming (LP) has accompanied IO analysis from its beginnings, and LP models based on IO tables are used to find production plans which maximize gross domestic product (GDP) or another objective subject to the balance equations and to capacity and resource constraints.

An LP problem attains its optimum at a vertex of the feasible set, where several capacity, demand and resource constraints are active. For a planner, this property has an unfavorable consequence: a small deviation of the optimal plan in a certain direction violates an active constraint, so that the optimal plan has no margin at all. In practice, plans are never implemented exactly, and IO data are estimates themselves. A planner therefore needs not only a good plan, but also tolerance ranges for outputs and final demands within which every plan remains acceptable, that is, keeps the economy balanced and keeps GDP close to its maximum. The purpose of this paper is to construct such plans and ranges and to explain which features of the economic structure determine them. Below, we review the related literature (Section~\ref{sec:litreview} and Table~\ref{tab:lit}), identify the research gap (Section~\ref{sec:gap}) and state our contributions (Section~\ref{sec:contrib}).

\subsection{Literature review}\label{sec:litreview}

\emph{Linear programming models based on IO tables.} The programming of interdependent activities was formulated by Wood and Dantzig \cite{WoodDantzig1949} and Dantzig \cite{Dantzig1951model}. In \cite{Koopmans1951,Koopmans1951efficient}, production has been considered as an efficient combination of activities, and the non-substitution theorem has been proved in \cite{Samuelson1951,Arrow1951}. The duality between Leontief quantities and prices was presented by Dorfman, Samuelson and Solow \cite{DSS1958}. Since then, LP models based on IO tables have been used in development planning \cite{CheneryKretschmer1956,Manne1974} and in the measurement of efficiency and productivity of national economies, where the observed economy is compared with its optimal LP frontier \cite{tenRaa2005,tenRaaMohnen2002,LuptacikBohm2010}. International trade has been analyzed in \cite{Duchin2005,DuchinLevine2011}, and waste, energy and emission problems in \cite{KondoNakamura2005,SanCristobal2010,HristuVarsakelis2010,Kang2020EST}. Multiobjective IO models describe the trade-offs between economic, energy, environmental and social objectives \cite{OliveiraAntunes2011,Carvalho2015,Oliveira2016}. The result of all these models is an optimal plan, a set of efficient plans or a system of shadow prices, and the plans are vertices or faces of the feasible set. How far such a plan can be changed without violating the balance equations and the target has not been examined in this literature.

\emph{Uncertainty of IO data.} It has long been recognized that economic observations, and IO data in particular, are not exact \cite{Morgenstern1951}. The effect of errors in IO coefficients on the solution of the IO model has been studied by stochastic analysis \cite{West1986}, Monte Carlo simulation \cite{BullardSebald1988} and interval arithmetic \cite{Jerrell1997}, and the bias of multipliers estimated from stochastic coefficients has been analyzed in \cite{Dietzenbacher2006,RuedaCantuche2013}. The important coefficients of an IO table have been identified by the analysis of coefficient change \cite{SonisHewings1992} and by the tolerable limits approach \cite{Tarancon2008,Lach2021,RamosCarvajal2024}. IO tables are often estimated or updated by biproportional methods \cite{LahrMesnard2004}, and the uncertainty of multiregional IO databases and of the carbon footprints computed from them has been assessed by Monte Carlo simulation and by the comparison of databases \cite{Lenzen2010,Wilting2012,MoranWood2014,Rodrigues2018}. This literature measures how errors in the data affect the results of IO analysis, but it does not construct plans which remain acceptable when the plan itself deviates.

\emph{Sensitivity analysis and robust optimization.} In LP, the stability of an optimal solution with respect to changes of the data is studied by sensitivity and parametric analysis \cite{WardWendell1990}. The tolerance approach determines the largest simultaneous and independent changes of the objective or right-hand side coefficients under which the optimal basis remains optimal \cite{Wendell1985,Filippi2005}, and a global version of this approach has been proposed in \cite{Borgonovo2018}. Robust optimization finds decisions which remain feasible for all realizations of the data in a given uncertainty set \cite{Soyster1973,BenTalNemirovski2000,BertsimasSim2004,BenTalElGhaouiNemirovski2009}; see \cite{BertsimasBrownCaramanis2011,Gabrel2014} for surveys. A robust IO--LP model for the electricity mix has been proposed in \cite{Kang2020EJOR}. Both approaches protect an optimal basis or a decision against perturbations of the data rather than describe how far the decision itself may deviate.

\emph{Near-optimal solution spaces.} Modeling to generate alternatives (MGA) searches for maximally different plans whose objective value is within a given slack of the optimum \cite{Brill1982,DeCarolis2011}. Near-optimal spaces have become important in energy system planning, since cost-optimal solutions often differ from the transitions observed in reality \cite{Trutnevyte2016}, and methods which explore the whole near-optimal space have been developed in \cite{NeumannBrown2021,Pedersen2021}. In IO analysis, near-optimal trade-offs between GDP and carbon dioxide emissions have been studied in \cite{Chang2024}. These methods describe the near-optimal set by alternative plans or by its extreme points, but they do not give a plan together with a guaranteed neighborhood in which every plan is both near-optimal and balanced.

\emph{Disruptions and resilience in IO models.} The inoperability IO model describes the propagation of a loss of functionality through interdependent infrastructures and sectors \cite{HaimesJiang2001,SantosHaimes2004}, adaptive regional and multiregional IO models describe the economic consequences of disasters \cite{Hallegatte2008,KoksThissen2016}, and an IO model with simultaneous supply and demand constraints has been used to assess the effects of the Covid-19 pandemic \cite{PichlerFarmer2022}; see \cite{OkuyamaSantos2014,GalbuseraGiannopoulos2018} for surveys. These models compute the consequences of a given shock and thus examine the boundary of the feasible set of an economy; the size of the interior of the set of acceptable plans is not measured.

\emph{Chebyshev centers, design centering and sphere packing.} The center of the largest ball inscribed into a convex set is its Chebyshev center, which is found by LP when the set is a polyhedron \cite{BoydVandenberghe2004,Balashov2022}; Chebyshev centers are also used in estimation with bounded errors \cite{EldarBeckTeboulle2008}. Related centers of polyhedra are the analytic center \cite{Sonnevend1986} and the center of the maximal inscribed ellipsoid \cite{KhachiyanTodd1993}. In engineering design, the design centering problem asks for a nominal design which keeps the largest possible set of perturbations within a region of acceptability \cite{VidigalDirector1982,NguyenStrodiot1992,Stein2006,HarwoodBarton2017}, and the maximal cube inscribed into a given domain has been constructed in \cite{Ashchepkov1980}. The sphere packing problem is to pack non-overlapping spheres with the maximum volume into a convex set. It has been shown to be equivalent to a convex maximization problem \cite{Enkhbat2020}; Malfatti's problem, a particular case of the circle packing problem, has been solved by a global optimization method \cite{Enkhbat2016}; and recent advances have been surveyed in \cite{Enkhbat2022MOTOR}. In the applications of the sphere packing approach, the center of the inscribed sphere gives a robust decision, and its radius shows how far the decision can be perturbed. This idea has been used in break-even and profitability analysis \cite{EnkhbatTungalag2023}, in financial management \cite{Chuluunbaatar2023}, in the government budget \cite{Ganlkhagva2023}, in taxation \cite{Badam2021} and in the mining industry \cite{Enkhbat2021NACO}, and it has been extended to nonconvex feasible--profit sets with a variable price in an oligopolistic market in \cite{Byambadash2026}. All these methods require a feasible set with a nonempty interior, which is defined by the inequalities of a single decision maker.

\begin{table}[H]
\caption{Positioning of the proposed approach in the related literature.}
\label{tab:lit}
\centering\footnotesize
\setlength{\tabcolsep}{3pt}
\begin{tabular}{>{\raggedright\arraybackslash}p{0.19\textwidth}>{\raggedright\arraybackslash}p{0.16\textwidth}>{\raggedright\arraybackslash}p{0.25\textwidth}>{\raggedright\arraybackslash}p{0.1\textwidth}>{\raggedright\arraybackslash}p{0.11\textwidth}>{\raggedright\arraybackslash}p{0.11\textwidth}}
\toprule
Stream & Studies & Object of the analysis & Balance equations & Guaranteed ranges of a plan & Link to IO structure\\
\midrule
LP--IO planning and efficiency & \cite{CheneryKretschmer1956,Manne1974,tenRaaMohnen2002,LuptacikBohm2010,Oliveira2016} & optimal or efficient plans, shadow prices & yes & no & shadow prices\\
Uncertainty of IO data & \cite{West1986,Dietzenbacher2006,Tarancon2008,Lenzen2010,Rodrigues2018} & effect of data errors on multipliers and footprints & yes & no & important coefficients\\
Sensitivity analysis, tolerance approach & \cite{Wendell1985,WardWendell1990,Filippi2005,Borgonovo2018} & stability of the optimal basis under changes of the data & yes & no & no\\
Robust optimization & \cite{BertsimasSim2004,BenTalElGhaouiNemirovski2009,Gabrel2014,Kang2020EJOR} & feasibility of a decision for all data in an uncertainty set & yes & no$^{a}$ & no\\
Near-optimal spaces, MGA & \cite{Brill1982,DeCarolis2011,NeumannBrown2021,Chang2024} & alternative near-optimal plans & yes & no & no\\
Disruption and inoperability IO models & \cite{HaimesJiang2001,Hallegatte2008,PichlerFarmer2022,GalbuseraGiannopoulos2018} & consequences of a given shock & yes & no & propagation of shocks\\
Chebyshev center, design centering, sphere packing & \cite{BoydVandenberghe2004,HarwoodBarton2017,Enkhbat2022MOTOR,EnkhbatTungalag2023} & largest ball in a set defined by inequalities & no$^{b}$ & yes & no\\
This paper & -- & largest ball in the $(\varepsilon,\delta)$-set of an LP--IO problem & yes & yes & sales concentration $\rho_i$, criticality weights\\
\bottomrule
\end{tabular}
\\[1mm]
\parbox[b]{0.95\linewidth}{\footnotesize Notes: $^{a}$ robust optimization guarantees feasibility for perturbations of the data, not for deviations of the plan; $^{b}$ the inscribed ball degenerates to a point when the constraints include equations (Proposition~\ref{prop:empty}).}
\end{table}

\subsection{Research gap}\label{sec:gap}

The reviewed streams answer different questions (Table~\ref{tab:lit}). LP--IO models find optimal or efficient plans; the literature on the uncertainty of IO data measures the effect of data errors on the results; sensitivity analysis and robust optimization protect an optimal basis or the feasibility of a decision against perturbations of the data; near-optimal space methods describe alternative plans; disruption models trace the consequences of a given shock; and the Chebyshev center and sphere packing approach measure how far a decision can move inside a set defined by inequalities. Taking these streams together, we identify the following three gaps.

\begin{enumerate}
\item[(G1)] \emph{The room for maneuver of a production plan is not measured.} None of the approaches characterizes a set of production plans in which every plan satisfies the supply--demand balances and a GDP target simultaneously, together with tolerance ranges which are guaranteed when all outputs and final demands deviate at the same time. The optimal plan of an LP--IO model has no such margin at all, and the near-optimal space methods describe alternative plans but not a guaranteed neighborhood of a plan.
\item[(G2)] \emph{The sphere packing approach does not apply to balance equations.} The Chebyshev center, design centering and sphere packing methods require a feasible set with a nonempty interior. The feasible set of a system of balance equations has an empty interior, and the inscribed sphere degenerates to a point. A relaxation of the balance equations which makes the approach applicable, and the properties of the resulting radius as a function of the tolerances, have not been studied.
\item[(G3)] \emph{The link between robustness and the structure of the economy is unknown.} It is not known which features of an IO table limit the margin of a balanced near-optimal plan, and how the sectors which limit this margin are related to the key sectors of the classical linkage analysis \cite{Rasmussen1956}. For Mongolia, an economy dominated by mining exports, recent IO and growth studies have identified key sectors, diversification priorities and the effects of the mining boom \cite{Dagys2025,Dagys2020,Lkhagva2019}, but the robustness of production plans of the economy has not been examined.
\end{enumerate}

\subsection{Contributions and organization}\label{sec:contrib}

In this paper, we relax the classical problem $\max\ip{c}{x}$ subject to $Ax=b$ in two directions: the objective function may be smaller than its optimal value $f^*$ by at most $\varepsilon$, and the balance may be violated by at most $\delta$ on both sides, that is, supply $S$ may differ from demand $D$ by at most $\delta$, $|S-D|\le\delta$. Then we find a sphere with the maximum radius inscribed into the obtained $(\varepsilon,\delta)$-set. The $\varepsilon$-relaxation connects the approach with the near-optimal spaces, while the $\delta$-relaxation, which is interpreted as the admissible change in inventories, converts the balance equations into inequalities; for positive tolerances $\varepsilon$ and $\delta$, the set of acceptable plans then has a nonempty interior. The main contributions of the paper are the following.
\begin{itemize}
\item \emph{Theory (gap G2).} We extend the sphere packing approach to linear programming problems with balance equations by introducing $(\varepsilon,\delta)$-solutions. We show that the $(\varepsilon,\delta)$-sphere packing problem reduces to linear programming, and we prove that the optimal radius is positive for positive tolerances, that it is non--decreasing, concave and piecewise linear with respect to the tolerances, that it is bounded by the narrowest balance constraint, and that every classical optimal solution has zero robustness whenever the objective function is not a linear combination of the balance equations, which is always the case in the IO model.
\item \emph{Method (gaps G1 and G3).} For the IO model, we write the sphere packing problem explicitly as a linear programming problem in relative coordinates, which gives a robust plan together with guaranteed tolerance ranges for all outputs and final demands. We obtain the bound $\delta/\rho_{\max}$ of the radius in terms of a sales-concentration index $\rho_i$ of sectors, the threshold of the GDP tolerance at which this bound is attained, and the criticality weights of the constraints, which relate the robustness of production plans to the structure of the economy.
\item \emph{Application (gaps G1 and G3).} Using the official IO tables of Mongolia for 2010--2019 and 2024, we show that the classical optimal plan has zero robustness, that a robust plan costs less than 0.3\% of the maximal GDP, and that the robustness is limited by sectors with concentrated sales, including the export-oriented mining sector, which differ from the key sectors of the linkage analysis.
\end{itemize}

The paper is organized as follows. The sphere packing problem is formulated in Section~\ref{sec:sp}. Section~\ref{sec:epsdelta} is devoted to the $(\varepsilon,\delta)$-sphere packing problem and its properties. Section~\ref{sec:io} applies the sphere packing approach to IO analysis, and Section~\ref{sec:numerical} presents the numerical results for the Mongolian economy. The results are discussed in relation to the research gaps in Section~\ref{sec:discussion}, and Section~\ref{sec:conclusion} concludes the paper.

\section{Formulation of the sphere packing problem}\label{sec:sp}

The main methodology of our research is the sphere packing approach, which deals with packing non-overlapping spheres with the maximum volume into a given set \cite{Enkhbat2020,Enkhbat2022MOTOR}. Let $B(x^0,r)$ be a ball with a center $x^0\in\R^n$ and radius $r\in\R$, $r\ge0$:
\begin{equation}\label{eq:ball}
B(x^0,r)=\{x\in\R^n\mid \norm{x-x^0}\le r\},
\end{equation}
here $\ip{\cdot}{\cdot}$ denotes the scalar product of two vectors in $\R^n$, and $\norm{\cdot}$ is the Euclidean norm. As is usual in the sphere packing approach, we also call the ball $B(x^0,r)$ a sphere. The $n$-dimensional volume of the Euclidean ball $B(x^0,r)$ is \cite{Enkhbat2020}
\begin{equation}\label{eq:vol}
V(B)=\frac{\pi^{n/2}}{\Gamma\bigl(\frac n2+1\bigr)}\,r^n,
\end{equation}
where $\Gamma$ is Euler's gamma function. Define the set $D$ as follows:
\begin{equation}\label{eq:Dconv}
D=\{x\in\R^n\mid g_i(x)\le0,\ i=\overline{1,m}\},
\end{equation}
where $g_i:\R^n\to\R$, $i=\overline{1,m}$, are convex functions. Assume that $D$ is a compact set and $\operatorname{int}D\ne\emptyset$. Clearly, $D$ is a convex set in $\R^n$. Let us introduce the functions $\varphi_i$ for $x\in\R^n$ and $r\ge0$:
\begin{equation}\label{eq:phi}
\varphi_i(x,r)=\max_{\norm{h}\le1}g_i(x+rh),\qquad i=\overline{1,m}.
\end{equation}
The maximum in \eqref{eq:phi} exists, since a convex function on $\R^n$ is continuous and the unit ball is compact. The following statement plays an important role in our further analysis.

\begin{lemma}[{[\citen{Enkhbat2020,Enkhbat2022MOTOR}]}]\label{lem:ball}
$B(x^0,r)\subset D$ if and only if
\begin{equation}\label{eq:lemcond}
\max_{1\le i\le m}\varphi_i(x^0,r)\le0.
\end{equation}
\end{lemma}
\begin{proof}
\emph{Necessity.} Let $B(x^0,r)\subset D$. Any point $y\in B(x^0,r)$ can be presented as $y=x^0+rh$, $h\in\R^n$, $\norm{h}\le1$. Then the condition $y\in D$ implies that
\[
g_i(x^0+rh)\le0\quad\text{for all}\ h\in\R^n,\ \norm{h}\le1,\qquad i=\overline{1,m}.
\]
Consequently, $\max_{1\le i\le m}\varphi_i(x^0,r)\le0$.

\emph{Sufficiency.} Let condition \eqref{eq:lemcond} be satisfied, and on the contrary, assume that there exists $y\in B(x^0,r)$ such that $y\notin D$. Then $y=x^0+rh$ for some $h\in\R^n$, $\norm{h}\le1$. Since $y\notin D$, there exists $j\in\{1,\dots,m\}$ such that $g_j(x^0+rh)>0$. Hence $\varphi_j(x^0,r)>0$, which contradicts \eqref{eq:lemcond}.
\end{proof}

Note that the proof of Lemma~\ref{lem:ball} uses neither the compactness of $D$ nor the condition $\operatorname{int}D\ne\emptyset$, so the lemma holds for every set of the form \eqref{eq:Dconv}; below, it is also applied to sets with an empty interior.

Let $D$ be a polytope given by the following linear inequalities:
\begin{equation}\label{eq:Dpoly}
D=\{x\in\R^n\mid \ip{a_i}{x}\le b_i,\ i=\overline{1,m}\},
\end{equation}
where $a_i\in\R^n$, $a_i\ne0$, and $b_i\in\R$. Then the functions $\varphi_i(x,r)$ in \eqref{eq:phi} are computed in the following way:
\begin{equation}\label{eq:phipoly}
\varphi_i(x,r)=\max_{\norm{h}\le1}\bigl[\ip{a_i}{x+rh}-b_i\bigr]=\ip{a_i}{x}+r\max_{\norm{h}\le1}\ip{a_i}{h}-b_i=\ip{a_i}{x}+r\norm{a_i}-b_i,\qquad i=\overline{1,m},
\end{equation}
since the maximum of $\ip{a_i}{h}$ over the unit ball is attained at $h=a_i/\norm{a_i}$. Consequently, by Lemma~\ref{lem:ball}, $B(x^0,r)\subset D$ if and only if
\begin{equation}\label{eq:ballcond}
\ip{a_i}{x^0}+r\norm{a_i}\le b_i,\qquad i=\overline{1,m}.
\end{equation}

The classical packing problem of inscribing $k$ equal spheres into $D$ with the maximum volume is \cite{Enkhbat2020,Enkhbat2022MOTOR}
\begin{align}
&\max_{(u,r)}\ V=\frac{k\,\pi^{n/2}}{\Gamma\bigl(\frac n2+1\bigr)}\,r^n\label{eq:kp_obj}\\
\intertext{subject to}
&\ip{a_i}{u^j}+r\norm{a_i}\le b_i,\qquad i=\overline{1,m},\ j=\overline{1,k},\label{eq:kp_in}\\
&\norm{u^i-u^j}^2\ge4r^2,\qquad i,j=\overline{1,k},\ i<j,\label{eq:kp_no}\\
&r\ge0,\label{eq:kp_r}
\end{align}
where $u^1,\dots,u^k$ are the centers of the spheres inscribed in $D$ and $r$ is their radius. Conditions \eqref{eq:kp_in} ensure that all spheres lie in $D$, while conditions \eqref{eq:kp_no} ensure that the spheres do not overlap. In what follows, we use the sphere packing problem with one sphere ($k=1$). Since $V$ is an increasing function of $r$, this problem is equivalent to the following linear programming problem:
\begin{align}
&\max_{(x,r)}\ r\label{eq:sp_obj}\\
\intertext{subject to}
&\ip{a_i}{x}+r\norm{a_i}\le b_i,\qquad i=\overline{1,m},\label{eq:sp_con}\\
&r\ge0.\label{eq:sp_r}
\end{align}
The center of a ball with the maximum radius contained in a set $D$ is called a Chebyshev center of $D$ \cite{BoydVandenberghe2004}; for the polytope \eqref{eq:Dpoly}, it is found from problem \eqref{eq:sp_obj}--\eqref{eq:sp_r}.

\section{The \texorpdfstring{$(\varepsilon,\delta)$}{(epsilon,delta)}-sphere packing problem}\label{sec:epsdelta}

The idea of this section is the following. A solution of a linear programming problem lies on the boundary of the feasible set, and if the constraints include equations, then the feasible set has no interior at all. In practice, however, a planner accepts plans whose objective value is slightly smaller than the optimal one and whose balances are violated slightly, for example, when small imbalances are absorbed by inventories. Therefore, we enlarge the feasible set by two tolerances, $\varepsilon$ for the objective function and $\delta$ for the balances, and look for the plan which is the farthest from the boundary of the enlarged set. The center of the largest sphere inscribed into this set is such a plan, and the radius of the sphere shows how far all variables may deviate simultaneously without violating any requirement.

\subsection{The classical problem and its \texorpdfstring{$(\varepsilon,\delta)$}{(epsilon,delta)}-set}

Consider the following linear programming problem:
\begin{align}
&\max\ f=\ip{c}{x}\label{eq:P_obj}\\
\intertext{subject to}
&Ax=b,\label{eq:P_eq}\\
&x\in X=\{x\in\R^n\mid Gx\le\gamma\},\label{eq:P_X}
\end{align}
where $A$ is an $m\times n$ matrix with rows $a_i$, $i=\overline{1,m}$, $G$ is a $K\times n$ matrix with rows $g_k$, $k=\overline{1,K}$, $c\in\R^n$, $b\in\R^m$ and $\gamma\in\R^K$. Assume that $c\ne0$, $a_i\ne0$ and $g_k\ne0$, that $X$ is a polytope with $\operatorname{int}X\ne\emptyset$, and that the feasible set $F=\{x\in X\mid Ax=b\}$ is nonempty. Denote by $x^*$ a solution to problem \eqref{eq:P_obj}--\eqref{eq:P_X} and by $f^*=\ip{c}{x^*}$ its optimal value. Problem \eqref{eq:P_obj}--\eqref{eq:P_X} will be called the classical problem. In IO analysis, the equality constraints \eqref{eq:P_eq} describe the balance between supply $S$ and demand $D$.

The following statement shows that the sphere packing approach cannot be applied to the classical problem directly.

\begin{proposition}\label{prop:empty}
If $m\ge1$, then $\operatorname{int}F=\emptyset$. Consequently, $B(x^0,r)\subset F$ implies $r=0$.
\end{proposition}
\begin{proof}
Let $x^0\in F$ and $r>0$. The point $y=x^0+r\,a_1/\norm{a_1}$ belongs to $B(x^0,r)$, but $\ip{a_1}{y}=b_1+r\norm{a_1}\ne b_1$. Hence $y\notin F$, and $B(x^0,r)$ is not contained in $F$.
\end{proof}

In order to overcome this difficulty, we relax the optimality and balance conditions. Let us introduce the following definition.

\begin{definition}\label{def:epsdelta}
Let $\varepsilon\ge0$ and $\delta=(\delta_1,\dots,\delta_m)\in\R^m_+$. The set
\begin{equation}\label{eq:Deps}
D(\varepsilon,\delta)=\{x\in X\mid \ip{c}{x}+\varepsilon\ge f^*,\ \ |\ip{a_i}{x}-b_i|\le\delta_i,\ i=\overline{1,m}\}
\end{equation}
is called the $(\varepsilon,\delta)$-set of problem \eqref{eq:P_obj}--\eqref{eq:P_X}, and its elements are called $(\varepsilon,\delta)$-solutions.
\end{definition}

Clearly, $D(\varepsilon,\delta)$ is a convex polyhedral set and $x^*\in D(\varepsilon,\delta)$. The vector $Ax-b$ measures the imbalance of the system; in the IO model, it is equal to $S-D$. It means that an $(\varepsilon,\delta)$-solution is a plan whose objective value is smaller than the optimal value by at most $\varepsilon$ and whose supply differs from demand by at most $\delta$ on both sides, i.e., $-\delta\le S-D\le\delta$. If $\varepsilon=0$ and $\delta=0$, then $D(\varepsilon,\delta)$ coincides with the set of solutions of the classical problem.

According to the sphere packing approach, we must find a ball $B(x^0,r)$ with the maximum radius $r$ such that $B(x^0,r)\subset D(\varepsilon,\delta)$:
\begin{equation}\label{eq:SP}
r^*(\varepsilon,\delta)=\max\{r\ge0\mid B(x^0,r)\subset D(\varepsilon,\delta)\ \text{for some}\ x^0\in\R^n\}.
\end{equation}
Problem \eqref{eq:SP} will be called the $(\varepsilon,\delta)$-sphere packing problem, and its solution $(x^0,r^*)$ gives the robust $(\varepsilon,\delta)$-plan $x^0$ and its radius $r^*$.

\begin{theorem}\label{thm:lp}
Problem \eqref{eq:SP} is equivalent to the following linear programming problem:
\begin{align}
&\max_{(x^0,r)}\ r\label{eq:LP_obj}\\
\intertext{subject to}
&-\ip{c}{x^0}+r\norm{c}\le\varepsilon-f^*,\label{eq:LP_gdp}\\
&\ip{a_i}{x^0}+r\norm{a_i}\le b_i+\delta_i,\qquad i=\overline{1,m},\label{eq:LP_balp}\\
&-\ip{a_i}{x^0}+r\norm{a_i}\le -b_i+\delta_i,\qquad i=\overline{1,m},\label{eq:LP_balm}\\
&\ip{g_k}{x^0}+r\norm{g_k}\le\gamma_k,\qquad k=\overline{1,K},\label{eq:LP_X}\\
&r\ge0.\label{eq:LP_r}
\end{align}
Problem \eqref{eq:LP_obj}--\eqref{eq:LP_r} has a solution for all $\varepsilon\ge0$ and $\delta\ge0$.
\end{theorem}
\begin{proof}
Clearly, $D(\varepsilon,\delta)$ is a polytope defined by the $1+2m+K$ linear inequalities $-\ip{c}{x}\le\varepsilon-f^*$, $\ip{a_i}{x}\le b_i+\delta_i$, $-\ip{a_i}{x}\le-b_i+\delta_i$, $i=\overline{1,m}$, and $\ip{g_k}{x}\le\gamma_k$, $k=\overline{1,K}$. By Lemma~\ref{lem:ball} and formula \eqref{eq:ballcond}, $B(x^0,r)\subset D(\varepsilon,\delta)$ if and only if conditions \eqref{eq:LP_gdp}--\eqref{eq:LP_r} hold. The point $(x^*,0)$ satisfies conditions \eqref{eq:LP_gdp}--\eqref{eq:LP_r}. Since $B(x^0,r)\subset X$ and $X$ is bounded, the radius $r$ is bounded on the feasible set of problem \eqref{eq:LP_obj}--\eqref{eq:LP_r}. Hence problem \eqref{eq:LP_obj}--\eqref{eq:LP_r} has a solution.
\end{proof}

\begin{example}\label{ex:toy}
Consider an economy with one sector whose input coefficient is $0.4$ and value added coefficient is $0.5$. Denote by $x$ and $y$ its output and final demand. The classical problem is to maximize $0.5x$ subject to $0.6x-y=0$ and $(x,y)\in X=[80,120]\times[45,66]$. The feasible set is the segment shown in Figure~\ref{fig:toy}(a), and the solution is $(110,66)$ with $f^*=55$. For $\varepsilon=5$ and $\delta=4$, the set $D(\varepsilon,\delta)$ is the quadrangle shown in Figure~\ref{fig:toy}(b). Solving problem \eqref{eq:LP_obj}--\eqref{eq:LP_r}, we find the optimal radius $r^*=\toyR$. Note that the ball with the maximum radius is not unique: its center can be any point $(x,0.6x)$ with $\toyCxMin\le x\le\toyCxMax$. Among these balls, the ball with the maximum GDP has the center $(\toyCx,\toyCy)$. This ball is shown in Figure~\ref{fig:toy}(b). Every plan of this ball is an $(\varepsilon,\delta)$-solution, while the classical solution lies on the boundary of $D(\varepsilon,\delta)$.
\end{example}

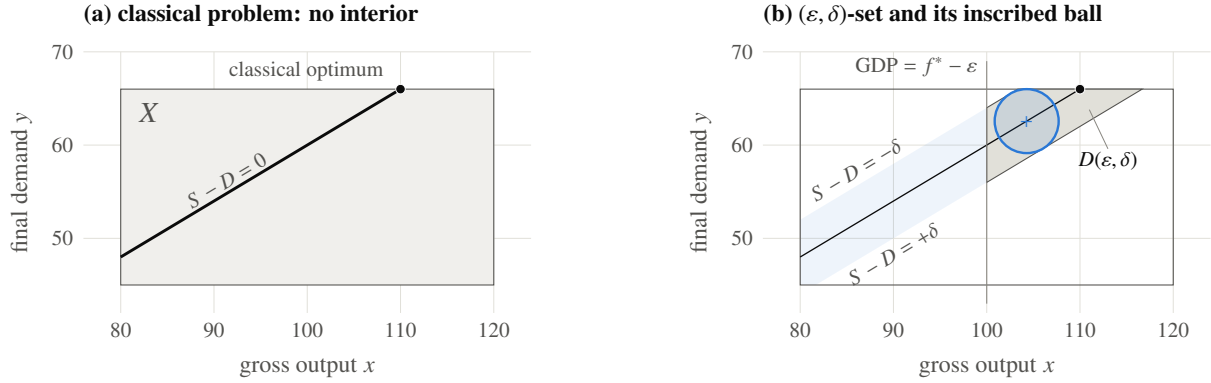
\begin{figure}[H]
\centering

\def\figOneCx{104.283381}
\def\figOneCy{62.570028}
\def\figOneR{3.429972}

\pgfplotsset{figonepanel/.style={paperfig, width=0.37\textwidth, axis equal image,
  xmin=76, xmax=124, ymin=42, ymax=70, xtick={80,90,100,110,120}, ytick={50,60,70},
  xlabel={gross output $x$}, ylabel={final demand $y$}}}%
\begin{tikzpicture}
\begin{axis}[figonepanel, title={(a) classical problem: no interior}]
\filldraw[fill=figWash, draw=figInkTwo, line width=0.4pt] (80,45) rectangle (120,66);
\addplot[figInk, line width=1.1pt] coordinates {(80,48) (110,66)}
  node[pos=0.4, sloped, above, figlabel] {$S-D=0$};
\addplot[only marks, figInk, figdot, mark size=1.8pt] coordinates {(110,66)};
\node[anchor=south east, figlabel] at (108.5,66.6) {classical optimum};
\node[text=figInkTwo] at (83,63.5) {$X$};
\end{axis}
\end{tikzpicture}\hfill
\begin{tikzpicture}
\begin{axis}[figonepanel, title={(b) $(\varepsilon,\delta)$-set and its inscribed ball}]
\addplot[fill=figBlue, fill opacity=0.08, draw=none] coordinates {(80.000000,52.000000) (80.000000,45.000000) (81.666667,45.000000) (116.666667,66.000000) (103.333333,66.000000)} -- cycle;
\draw[figInkTwo, line width=0.4pt] (80,45) rectangle (120,66);
\addplot[fill=figGrid, draw=figInkTwo, line width=0.4pt] coordinates {(100.000000,56.000000) (116.666667,66.000000) (103.333333,66.000000) (100.000000,64.000000)} -- cycle;
\addplot[figInk, line width=0.5pt] coordinates {(80,48) (110,66)};
\filldraw[fill=figBlue, fill opacity=0.12, draw=figBlue, line width=0.9pt]
  (\figOneCx,\figOneCy) circle[radius=\figOneR];
\node[text=figBlue, font=\scriptsize] at (\figOneCx,\figOneCy) {$+$};
\addplot[only marks, figInk, figdot, mark size=1.8pt] coordinates {(110,66)};
\draw[figMuted, line width=0.4pt] (100,43) -- (100,69);
\node[anchor=east, figlabel] at (99.3,68.4) {GDP $=f^*-\varepsilon$};
\node[rotate=30.96, figlabel] at (86,57.3) {$S-D=-\delta$};
\node[rotate=30.96, figlabel] at (90,48.3) {$S-D=+\delta$};
\draw[figMuted, line width=0.3pt] (111,63.8) -- (113,59.6);
\node[anchor=north, figlabel, text=figInk] at (113,59.6) {$D(\varepsilon,\delta)$};
\end{axis}
\end{tikzpicture}
\caption{(a) The feasible set of the classical problem in Example~\ref{ex:toy}. (b) The $(\varepsilon,\delta)$-set and the ball with the maximum radius inscribed into it.}
\label{fig:toy}
\end{figure}

\subsection{Properties of the \texorpdfstring{$(\varepsilon,\delta)$}{(epsilon,delta)}-sphere packing problem}

\begin{theorem}\label{thm:pos}
If $\varepsilon>0$ and $\delta_i>0$, $i=\overline{1,m}$, then $r^*(\varepsilon,\delta)>0$. If $m\ge1$ and $\delta=0$, then $r^*(\varepsilon,\delta)=0$ for all $\varepsilon\ge0$.
\end{theorem}
\begin{proof}
Take any $\hat x\in\operatorname{int}X$. Since $g_k\ne0$, the point $\hat x+s\,g_k$ belongs to $X$ for a sufficiently small $s>0$, and hence $\ip{g_k}{\hat x}\le\gamma_k-s\norm{g_k}^2<\gamma_k$, $k=\overline{1,K}$, i.e., $G\hat x<\gamma$. Consider the points $x_t=(1-t)x^*+t\hat x$, $t\in(0,1]$. Then $Gx_t=(1-t)Gx^*+tG\hat x<\gamma$. Moreover, $Ax_t-b=t(A\hat x-b)$ and $\ip{c}{x_t}-f^*=t(\ip{c}{\hat x}-f^*)$. Therefore, for a sufficiently small $t$, we have $t|\ip{a_i}{\hat x}-b_i|<\delta_i$, $i=\overline{1,m}$, and $t(\ip{c}{\hat x}-f^*)>-\varepsilon$. It means that all inequalities defining $D(\varepsilon,\delta)$ hold strictly at $x_t$. Hence $x_t\in\operatorname{int}D(\varepsilon,\delta)$ and $r^*(\varepsilon,\delta)>0$. If $\delta=0$, then $D(\varepsilon,\delta)\subset F$, and the assertion follows from Proposition~\ref{prop:empty}.
\end{proof}

\begin{remark}\label{rem:eps0}
The condition $\varepsilon>0$ cannot be omitted in general. Indeed, consider the problem $\max x_1$ subject to $x_1-x_2=0$ and $x\in[0,1]^2$. Then $f^*=1$, the set $D(0,\delta)=\{1\}\times[\max\{0,1-\delta\},1]$ is a segment, and $r^*(0,\delta)=0$ for all $\delta\ge0$. More precisely, if $\delta_i>0$, $i=\overline{1,m}$, then $r^*(0,\delta)>0$ if and only if $\ip{c}{x}>f^*$ for some $x\in X$, i.e., if the balance equations restrict the optimal value. Indeed, if $B(x^0,r)\subset D(0,\delta)$ with $r>0$, then the point $x^0+r\,c/\norm{c}$ belongs to $X$ and $\ip{c}{x^0+r\,c/\norm{c}}\ge f^*+r\norm{c}>f^*$. Conversely, if $\ip{c}{x}>f^*$ for some $x\in X$, then a point $\hat x\in\operatorname{int}X$ close to $x$ satisfies $\ip{c}{\hat x}>f^*$, and the proof of Theorem~\ref{thm:pos} with this point $\hat x$ applies to $\varepsilon=0$. In the above example, $\ip{c}{x}\le f^*$ for all $x\in X$, whereas $r^*(0,\delta)>0$ in the IO model considered below (Section~\ref{sec:tradeoff}).
\end{remark}

\begin{theorem}\label{thm:concave}
The optimal radius $r^*(\varepsilon,\delta)$ is non--decreasing, concave and piecewise linear on $\R_+\times\R^m_+$.
\end{theorem}
\begin{proof}
If $\varepsilon\le\varepsilon'$ and $\delta\le\delta'$, then $D(\varepsilon,\delta)\subset D(\varepsilon',\delta')$, which implies monotonicity. The parameters $\vartheta=(\varepsilon,\delta)$ enter only the right-hand side of problem \eqref{eq:LP_obj}--\eqref{eq:LP_r}, and they enter it linearly. Let $(x^1,r^1)$ and $(x^2,r^2)$ be solutions to problem \eqref{eq:LP_obj}--\eqref{eq:LP_r} for $\vartheta^1$ and $\vartheta^2$, and let $\lambda\in[0,1]$. Then the point $\lambda(x^1,r^1)+(1-\lambda)(x^2,r^2)$ is feasible for $\lambda\vartheta^1+(1-\lambda)\vartheta^2$. Consequently, $r^*(\lambda\vartheta^1+(1-\lambda)\vartheta^2)\ge\lambda r^1+(1-\lambda)r^2$, which proves concavity. For piecewise linearity, write constraints \eqref{eq:LP_gdp}--\eqref{eq:LP_X} in the normalized form $\ip{\hat p_k}{x^0}+r\le\hat q_k(\vartheta)$ of Proposition~\ref{prop:dual} below. The set $\Lambda=\{\lambda\in\R^N_+\mid\sum_k\lambda_k\hat p_k=0,\ \sum_k\lambda_k=1\}$ does not depend on $\vartheta$, and by the duality theorem of linear programming \cite{Dantzig1963}, $r^*(\vartheta)=\min_{\lambda\in\Lambda}\sum_k\lambda_k\hat q_k(\vartheta)$, since $r^*(\vartheta)\ge0$. The minimum is attained at a basic point of $\Lambda$, i.e., at a point $\lambda$ for which the vectors $(\hat p_k,1)$ with $\lambda_k>0$ are linearly independent. Such a point is determined by the set $\{k\mid\lambda_k>0\}$, so $\Lambda$ has finitely many basic points. Hence $r^*$ is the minimum of finitely many affine functions of $\vartheta$, i.e., it is piecewise linear.
\end{proof}

\begin{proposition}\label{prop:slab}
Let $m\ge1$. For all $\varepsilon\ge0$, the following inequality holds:
\begin{equation}\label{eq:slab}
r^*(\varepsilon,\delta)\le\min_{1\le i\le m}\frac{\delta_i}{\norm{a_i}}.
\end{equation}
In particular, $r^*(\varepsilon,\delta)\to0$ as $\delta\to0$ uniformly in $\varepsilon$.
\end{proposition}
\begin{proof}
Adding inequalities \eqref{eq:LP_balp} and \eqref{eq:LP_balm} for the $i$-th row, we obtain $2r\norm{a_i}\le2\delta_i$, $i=\overline{1,m}$.
\end{proof}

Inequality \eqref{eq:slab} shows the rigidity of balance systems. It means that, however large the tolerance $\varepsilon$ is, the margin of a plan is restricted by the balance tolerance, and the most restrictive balance is the narrowest one.

\begin{proposition}\label{prop:fragile}
Assume that $c$ is not a linear combination of the vectors $a_1,\dots,a_m$. Then every solution $x^*$ of problem \eqref{eq:P_obj}--\eqref{eq:P_X} satisfies $\ip{g_k}{x^*}=\gamma_k$ for some $k$. Consequently, for all $\varepsilon$ and $\delta$, no ball with a positive radius centered at $x^*$ is contained in $D(\varepsilon,\delta)$.
\end{proposition}
\begin{proof}
Write $c=A^\top\nu+d$, where $d$ is the orthogonal projection of $c$ onto the null space of $A$. By the assumption, $d\ne0$. Assume that $Gx^*<\gamma$. Then for a sufficiently small $t>0$, the point $x^*+td$ belongs to $X$ and satisfies $A(x^*+td)=b$. Moreover, $\ip{c}{x^*+td}=f^*+t\norm{d}^2>f^*$, which contradicts the optimality of $x^*$. Hence $\ip{g_k}{x^*}=\gamma_k$ for some $k$, and the point $x^*+r\,g_k/\norm{g_k}\in B(x^*,r)$ does not belong to $X$ for any $r>0$.
\end{proof}

\begin{theorem}\label{thm:sym}
Let $D$ be a nonempty convex compact set which is centrally symmetric about a point $p$, i.e., $x\in D$ implies $2p-x\in D$. Then $p$ is a Chebyshev center of $D$.
\end{theorem}
\begin{proof}
Let $B(x^0,r)\subset D$. By the symmetry, $B(2p-x^0,r)\subset D$. For any $h$ with $\norm{h}\le r$, we have $p+h=\frac12(x^0+h)+\frac12(2p-x^0+h)\in D$ by the convexity of $D$. Hence $B(p,r)\subset D$. In particular, $p\in D$ (take $x^0\in D$ and $r=0$), and the radii $r$ of the balls $B(p,r)\subset D$ are bounded, since $D$ is bounded. Let $R$ be their supremum. If $R>0$, then the open ball with the center $p$ and radius $R$ lies in $D$, and since $D$ is closed, $B(p,R)\subset D$; if $R=0$, then $B(p,R)=\{p\}\subset D$. Thus $B(p,R)$ is a ball with the maximum radius in $D$, and $p$ is a Chebyshev center of $D$.
\end{proof}

\begin{proposition}\label{prop:box}
Let $(x^0,r^*)$ be a solution to problem \eqref{eq:LP_obj}--\eqref{eq:LP_r}. Then every point $x$ satisfying
\begin{equation}\label{eq:intervals}
x_j^0-\frac{r^*}{\sqrt n}\le x_j\le x_j^0+\frac{r^*}{\sqrt n},\qquad j=\overline{1,n},
\end{equation}
is an $(\varepsilon,\delta)$-solution. Moreover, the largest cube $\{x\mid \norm{x-u}_\infty\le t\}$ contained in $D(\varepsilon,\delta)$ is found from problem \eqref{eq:LP_obj}--\eqref{eq:LP_r} with the variables $(u,t)$ instead of $(x^0,r)$ by replacing every Euclidean norm $\norm{\cdot}$ with the norm $\norm{\cdot}_1$, and its half-width satisfies $t^*\ge r^*/\sqrt n$.
\end{proposition}
\begin{proof}
If $x$ satisfies \eqref{eq:intervals}, then $\norm{x-x^0}\le\sqrt n\max_j|x_j-x_j^0|\le r^*$, so $x\in B(x^0,r^*)\subset D(\varepsilon,\delta)$. For the cube, H\"older's inequality gives $\max_{\norm{h}_\infty\le1}\ip{p}{h}=\norm{p}_1$, and the proof of Lemma~\ref{lem:ball} applies to the cube as well. As in the proof of Theorem~\ref{thm:lp}, the obtained problem has a solution $(u^*,t^*)$, since the point $(x^*,0)$ satisfies its constraints and $t$ is bounded on its feasible set, because the cube lies in the bounded set $X$. Since the cube with the center $x^0$ and the half-width $r^*/\sqrt n$ lies in $B(x^0,r^*)$, we have $t^*\ge r^*/\sqrt n$.
\end{proof}

It means that a plan is guaranteed to remain an $(\varepsilon,\delta)$-solution if every variable is kept in intervals \eqref{eq:intervals}, even if all variables deviate simultaneously. The condition is sufficient, but not necessary: a plan outside these intervals can still be an $(\varepsilon,\delta)$-solution, for instance, if it lies in the ball $B(x^0,r^*)$. Similar intervals have been used in break-even analysis \cite{EnkhbatTungalag2023}, and the cube corresponds to the survival theory \cite{Ashchepkov1980,Badam2021}.

\begin{proposition}\label{prop:dual}
Denote by $\ip{p_k}{x}\le q_k$, $k=\overline{1,N}$, $N=1+2m+K$, the inequalities which define $D(\varepsilon,\delta)$, and write constraints \eqref{eq:LP_gdp}--\eqref{eq:LP_X} in the normalized form $\ip{\hat p_k}{x^0}+r\le\hat q_k$, where $\hat p_k=p_k/\norm{p_k}$ and $\hat q_k=q_k/\norm{p_k}$. Then the dual problem of \eqref{eq:LP_obj}--\eqref{eq:LP_r} is
\begin{equation}\label{eq:dual}
\min_{\lambda\ge0}\ \Bigl\{\sum_{k=1}^N\lambda_k\hat q_k\ \Bigm|\ \sum_{k=1}^N\lambda_k\hat p_k=0,\ \sum_{k=1}^N\lambda_k\ge1\Bigr\}.
\end{equation}
Problem \eqref{eq:dual} has a solution, and its optimal value is $r^*$. Every dual solution $\lambda^*$ satisfies $\lambda_k^*=0$ for the constraints which are inactive at a solution $(x^0,r^*)$ of problem \eqref{eq:LP_obj}--\eqref{eq:LP_r}, and $\sum_k\lambda_k^*=1$ if $r^*>0$. If, moreover, the dual solution $\lambda^*$ is unique (for instance, if $r^*>0$ and the optimal basis is nondegenerate), then $\partial r^*/\partial\hat q_k=\lambda_k^*$.
\end{proposition}
\begin{proof}
The existence of a dual solution and the equality of the optimal values follow from the duality theorem of linear programming, since problem \eqref{eq:LP_obj}--\eqref{eq:LP_r} has a solution by Theorem~\ref{thm:lp}. The dual constraint corresponding to the variable $r\ge0$ is $\sum_k\lambda_k\ge1$, and by the complementary slackness condition it holds as an equality if $r^*>0$; the same condition gives $\lambda_k^*=0$ for the constraints which are inactive at $(x^0,r^*)$. If $\lambda^*$ is unique, then $r^*>0$, since for $r^*=0$ the point $2\lambda^*$ would be another solution of problem \eqref{eq:dual}, and hence $\sum_k\lambda_k^*=1$. By the proof of Theorem~\ref{thm:concave}, $r^*$ is the minimum of the finitely many linear functions $\sum_k\lambda_k\hat q_k$ of $\hat q$, where $\lambda$ runs over the basic points of $\Lambda$, and $\lambda^*$ is the only basic point at which this minimum is attained. Hence $r^*=\sum_k\lambda_k^*\hat q_k$ in a neighborhood of $\hat q$, which gives $\partial r^*/\partial\hat q_k=\lambda_k^*$. Finally, let $r^*>0$ and let the optimal basis be nondegenerate, i.e., exactly $n+1$ of constraints \eqref{eq:LP_gdp}--\eqref{eq:LP_X} are active at $(x^0,r^*)$, and their vectors $(\hat p_k,1)$ are linearly independent. Every dual solution vanishes on the inactive constraints and satisfies $\sum_k\lambda_k(\hat p_k,1)=(0,1)$, and these equations have at most one solution supported by the active constraints, so the dual solution is unique.
\end{proof}

We call $\lambda^*$ the criticality weights. If $r^*>0$, they sum to one and distribute the responsibility for the limited radius among the constraints. Since $\sum_k\lambda_k^*\hat p_k=0$, the normals of the active constraints surround the center of the ball, and no small shift of the center can increase the radius.

\section{Input--output analysis}\label{sec:io}

In this section, we show how to apply the sphere packing approach to input--output analysis.

\subsection{Input--output model}\label{sec:iomodel}

Consider an economy with $n$ sectors. Denote by $z_{ij}$ the intermediate delivery from sector $i$ to sector $j$, by $\bar x_i$ the gross output of sector $i$ and by $\bar y_i=\bar x_i-\sum_{j=1}^nz_{ij}$ its final demand in the observed input--output table. The technical coefficients $a_{ij}=z_{ij}/\bar x_j$, $i,j=\overline{1,n}$, form the matrix $A$. Denote by $v_j$, $m_j$ and $w_j$ the value added, the imported inputs and the compensation of employees per unit of output of sector $j$, and let $\bar V=\ip{v}{\bar x}$, $\bar M=\ip{m}{\bar x}$ and $\bar W=\ip{w}{\bar x}$ be the observed GDP, imported inputs and wage bill, respectively. Assume that $z_{ij}\ge0$, $i,j=\overline{1,n}$, that $\bar x_i>0$ and $\bar y_i>0$, $i=\overline{1,n}$, that $m_j\ge0$ and $w_j\ge0$, $j=\overline{1,n}$, and that $\bar V>0$, $\bar M>0$ and $\bar W>0$. These conditions hold for all input--output tables used in Section~\ref{sec:numerical}: the smallest share of final demand in the output of a sector is \minShareFD\ (\minShareFDsector, \minShareFDyear). If, moreover, every column sum $\sum_{i=1}^na_{ij}$ is less than one, i.e., the intermediate inputs of every sector are smaller than its output, then every eigenvalue of $A$ has modulus at most $\max_j\sum_{i=1}^na_{ij}<1$, and the Leontief inverse $(I-A)^{-1}=\sum_{k=0}^\infty A^k$ exists and is nonnegative \cite{HawkinsSimon1949}.

A plan of the economy consists of gross outputs $x=(x_1,\dots,x_n)$ and final demands $y=(y_1,\dots,y_n)$. The plan must satisfy the balance between the supply $S_i=x_i$ of sector $i$ and the demand $D_i=\sum_{j=1}^na_{ij}x_j+y_i$ for its output. Moreover, outputs and final demands must stay within operating corridors around their observed levels, which reflect capacities and minimal operating levels, and the use of imported inputs and labor is limited by the available foreign exchange and labor force. Then the classical LP--IO problem, which maximizes GDP, is written as follows:
\begin{align}
&\max_{(x,y)}\ \ip{v}{x}\label{eq:IO_obj}\\
\intertext{subject to}
&(I-A)x-y=0,\label{eq:IO_bal}\\
&(1-\kappa)\bar x\le x\le(1+\kappa)\bar x,\qquad (1-\kappa)\bar y\le y\le(1+\kappa)\bar y,\label{eq:IO_corr}\\
&\ip{m}{x}\le(1+\theta)\bar M,\qquad \ip{w}{x}\le(1+\theta)\bar W,\label{eq:IO_res}
\end{align}
where $\kappa$ is the half-width of the corridors and $\theta$ is the headroom for imported inputs and the wage bill. Problem \eqref{eq:IO_obj}--\eqref{eq:IO_res} is a particular case of problem \eqref{eq:P_obj}--\eqref{eq:P_X} with the variables $\xi=(x,y)\in\R^{2n}$, the objective vector $c=(v,0)$, the balance matrix $[I-A,\,-I]$, the right-hand side $b=0$ and the set $X$ defined by \eqref{eq:IO_corr}--\eqref{eq:IO_res}. The last $n$ components of any linear combination $\sum_{i=1}^n\nu_i\bigl((I-A)_i,-e_i\bigr)$ of the rows of the balance matrix are equal to $-\nu$. Therefore, the objective vector $c=(v,0)$ with $v\ne0$ (recall that $\bar V>0$) is not a linear combination of these rows, and Proposition~\ref{prop:fragile} applies to problem \eqref{eq:IO_obj}--\eqref{eq:IO_res}: every classical optimal plan lies on the boundary of $X$. The same argument holds in the relative coordinates introduced in Section~\ref{sec:iosp}.

In practice, small imbalances between supply and demand are absorbed by changes in inventories, and a plan with GDP slightly smaller than the maximal one is still acceptable. Therefore, we introduce the $(\varepsilon,\delta)$-set of problem \eqref{eq:IO_obj}--\eqref{eq:IO_res}:
\begin{equation}\label{eq:DIO}
D_{IO}(\varepsilon,\delta)=\{(x,y)\in X\mid \ip{v}{x}+\varepsilon\ge f^*,\ \ |S_i-D_i|\le\delta_i,\ i=\overline{1,n}\},
\end{equation}
where $S_i-D_i=x_i-\sum_{j=1}^na_{ij}x_j-y_i$ is the imbalance of sector $i$, the balance tolerance is measured relative to output, $\delta_i=\delta\,\bar x_i$, $i=\overline{1,n}$, and the GDP tolerance relative to the optimal value, $\varepsilon=\varepsilon_{\rm rel}f^*$.

\subsection{Sphere packing problem for the input--output model}\label{sec:iosp}

Note that the outputs and final demands of sectors differ by more than two orders of magnitude (from 0.22 to 47.9 trillion MNT in 2024), and a ball measured in monetary units is dominated by the smallest sectors (see Section~\ref{sec:sensitivity}). For this reason, we measure every output and every final demand relative to its observed level. Introduce the relative coordinates $\zeta=(\zeta_1,\dots,\zeta_{2n})\in\R^{2n}$ by
\begin{equation}\label{eq:relcoord}
x_j=\bar x_j(1+\zeta_j),\qquad y_j=\bar y_j(1+\zeta_{n+j}),\qquad j=\overline{1,n}.
\end{equation}
Then $\zeta_j$ and $\zeta_{n+j}$ are the relative deviations of the output and the final demand of sector $j$ from the observed table, and the observed table corresponds to $\zeta=0$. Denote by $B(\zeta^0,r)$ a sphere with a center $\zeta^0\in\R^{2n}$ and radius $r\ge0$:
\begin{equation}\label{eq:ballIO}
B(\zeta^0,r)=\{\zeta\in\R^{2n}\mid \norm{\zeta-\zeta^0}\le r\}.
\end{equation}
In the original units, $B(\zeta^0,r)$ is the ellipsoid
\begin{equation}\label{eq:ellipsoid}
\sum_{j=1}^n\Bigl(\frac{x_j-x_j^0}{\bar x_j}\Bigr)^2+\sum_{j=1}^n\Bigl(\frac{y_j-y_j^0}{\bar y_j}\Bigr)^2\le r^2
\end{equation}
with the center $x_j^0=\bar x_j(1+\zeta_j^0)$, $y_j^0=\bar y_j(1+\zeta_{n+j}^0)$, $j=\overline{1,n}$. Thus, the radius $r$ is a unit-free measure of joint relative deviations of all outputs and final demands. Below, all radii are given in percent.

Let us introduce the following coefficients of the observed table:
\begin{gather}
b_{ij}=\frac{z_{ij}}{\bar x_i},\qquad s_i=\frac{\bar y_i}{\bar x_i},\qquad \rho_i=\Bigl((1-b_{ii})^2+\sum_{j\ne i}b_{ij}^2+s_i^2\Bigr)^{1/2},\qquad i,j=\overline{1,n},\label{eq:rho}\\
\sigma_j=\frac{v_j\bar x_j}{\bar V},\qquad \mu_j=\frac{m_j\bar x_j}{\bar M},\qquad \omega_j=\frac{w_j\bar x_j}{\bar W},\qquad j=\overline{1,n}.\label{eq:shares}
\end{gather}
Here $b_{ij}$ is the share of the output of sector $i$ sold to sector $j$, $s_i$ is the share sold to final demand, and $\sigma_j$, $\mu_j$ and $\omega_j$ are the shares of sector $j$ in GDP, imported inputs and the wage bill of the observed table, respectively. Denote also by
\begin{equation}\label{eq:eta}
\eta=\frac{f^*-\varepsilon}{\bar V}-1
\end{equation}
the minimal growth of GDP over its observed level, which is required in the set $D_{IO}(\varepsilon,\delta)$.

According to the sphere packing approach, we must find a sphere $B(\zeta^0,r)$ with the maximum radius $r$ such that every plan of this sphere belongs to $D_{IO}(\varepsilon,\delta)$. For this purpose, we rewrite problem \eqref{eq:LP_obj}--\eqref{eq:LP_r} taking into account the constraints of $D_{IO}(\varepsilon,\delta)$ given by \eqref{eq:DIO} as the following linear programming problem:
\begin{align}
&\max_{(\zeta^0,r)}\ r\label{eq:IOSP_obj}\\
\intertext{subject to}
&\sum_{j=1}^n\sigma_j\zeta_j^0-r\norm{\sigma}\ge\eta,\label{eq:IOSP_gdp}\\
&\zeta_i^0-\sum_{j=1}^nb_{ij}\zeta_j^0-s_i\zeta_{n+i}^0+\rho_i\,r\le\delta,\qquad i=\overline{1,n},\label{eq:IOSP_sur}\\
&-\zeta_i^0+\sum_{j=1}^nb_{ij}\zeta_j^0+s_i\zeta_{n+i}^0+\rho_i\,r\le\delta,\qquad i=\overline{1,n},\label{eq:IOSP_sho}\\
&\zeta_j^0+r\le\kappa,\qquad -\zeta_j^0+r\le\kappa,\qquad j=\overline{1,2n},\label{eq:IOSP_corr}\\
&\sum_{j=1}^n\mu_j\zeta_j^0+r\norm{\mu}\le\theta,\qquad \sum_{j=1}^n\omega_j\zeta_j^0+r\norm{\omega}\le\theta,\label{eq:IOSP_res}\\
&r\ge0.\label{eq:IOSP_r}
\end{align}
Constraint \eqref{eq:IOSP_gdp} is the GDP requirement, constraints \eqref{eq:IOSP_sur} and \eqref{eq:IOSP_sho} are the surplus side $S_i-D_i\le\delta\bar x_i$ and the shortage side $D_i-S_i\le\delta\bar x_i$ of the supply--demand balances, constraints \eqref{eq:IOSP_corr} are the corridors for outputs and final demands, and constraints \eqref{eq:IOSP_res} restrict imported inputs and the wage bill. Problem \eqref{eq:IOSP_obj}--\eqref{eq:IOSP_r} will be called the sphere packing problem for the input--output model. It has $2n+1$ variables, $6n+3$ linear inequality constraints \eqref{eq:IOSP_gdp}--\eqref{eq:IOSP_res} and the sign constraint $r\ge0$.

\begin{theorem}\label{thm:io}
In the relative coordinates \eqref{eq:relcoord}, $B(\zeta^0,r)\subset D_{IO}(\varepsilon,\delta)$ if and only if $(\zeta^0,r)$ satisfies constraints \eqref{eq:IOSP_gdp}--\eqref{eq:IOSP_r}. Consequently, problem \eqref{eq:IOSP_obj}--\eqref{eq:IOSP_r} is the $(\varepsilon,\delta)$-sphere packing problem \eqref{eq:SP} for the input--output model \eqref{eq:IO_obj}--\eqref{eq:IO_res} written in the relative coordinates \eqref{eq:relcoord}.
\end{theorem}
\begin{proof}
We rewrite every constraint of $D_{IO}(\varepsilon,\delta)$ in the relative coordinates \eqref{eq:relcoord} and apply formula \eqref{eq:ballcond}. Note that dividing a linear inequality by a positive number does not change condition \eqref{eq:ballcond}. Since the observed table is balanced, $\bar x_i-\sum_{j=1}^nz_{ij}-\bar y_i=0$, and the imbalance of sector $i$ takes the form
\[
S_i-D_i=\bar x_i\zeta_i-\sum_{j=1}^nz_{ij}\zeta_j-\bar y_i\zeta_{n+i}=\bar x_i\Bigl(\zeta_i-\sum_{j=1}^nb_{ij}\zeta_j-s_i\zeta_{n+i}\Bigr).
\]
Therefore, the balance condition $|S_i-D_i|\le\delta\bar x_i$ is equivalent to $|\zeta_i-\sum_{j=1}^nb_{ij}\zeta_j-s_i\zeta_{n+i}|\le\delta$. The normal vector of this linear function has the norm $\rho_i$, and formula \eqref{eq:ballcond} gives constraints \eqref{eq:IOSP_sur} and \eqref{eq:IOSP_sho}. The GDP condition $\ip{v}{x}\ge f^*-\varepsilon$ takes the form $\bar V\bigl(1+\sum_{j=1}^n\sigma_j\zeta_j\bigr)\ge f^*-\varepsilon$, i.e., $\sum_{j=1}^n\sigma_j\zeta_j\ge\eta$, which gives \eqref{eq:IOSP_gdp}. The corridors \eqref{eq:IO_corr} take the form $-\kappa\le\zeta_j\le\kappa$, $j=\overline{1,2n}$, and their normals have the unit norm, which gives \eqref{eq:IOSP_corr}. Finally, the condition $\ip{m}{x}\le(1+\theta)\bar M$ is equivalent to $\sum_{j=1}^n\mu_j\zeta_j\le\theta$, and the wage condition is considered in the same way, which gives \eqref{eq:IOSP_res}.
\end{proof}

In other words, in the relative coordinates, problem \eqref{eq:IO_obj}--\eqref{eq:IO_res} is the problem of maximizing $\sum_{j=1}^n\sigma_j\zeta_j$ subject to the balance equations $\zeta_i-\sum_{j=1}^nb_{ij}\zeta_j-s_i\zeta_{n+i}=0$, $i=\overline{1,n}$, and $\zeta\in X_\zeta=\{\zeta\mid|\zeta_k|\le\kappa,\ k=\overline{1,2n},\ \sum_{j=1}^n\mu_j\zeta_j\le\theta,\ \sum_{j=1}^n\omega_j\zeta_j\le\theta\}$, with the optimal value $(f^*-\bar V)/\bar V$, and problem \eqref{eq:IOSP_obj}--\eqref{eq:IOSP_r} is problem \eqref{eq:LP_obj}--\eqref{eq:LP_r} for this problem with the GDP tolerance $\varepsilon/\bar V$ and the balance tolerances $\delta_i=\delta$. The assumptions of Section~\ref{sec:epsdelta} hold: the normal vectors are nonzero, $X_\zeta$ is bounded, the observed table $\zeta=0$ is feasible, and, if $\kappa>0$, then $\operatorname{int}X_\zeta\ne\emptyset$, since the point with $\zeta_k=-\kappa/2$, $k=\overline{1,2n}$, satisfies all inequalities of $X_\zeta$ strictly. Therefore, Theorems~\ref{thm:lp} and~\ref{thm:concave} and Propositions~\ref{prop:slab}--\ref{prop:dual} hold for the optimal radius $r^*(\varepsilon,\delta)$ of problem \eqref{eq:IOSP_obj}--\eqref{eq:IOSP_r}, and so does Theorem~\ref{thm:pos}, whose first assertion requires $\kappa>0$ (if $\kappa=0$, then $X_\zeta=\{0\}$ and $r^*=0$).

Let $(\zeta^0,r^*)$ be a solution to problem \eqref{eq:IOSP_obj}--\eqref{eq:IOSP_r}. Then the robust plan is $x_j^0=\bar x_j(1+\zeta_j^0)$, $y_j^0=\bar y_j(1+\zeta_{n+j}^0)$, $j=\overline{1,n}$, and every plan of ellipsoid \eqref{eq:ellipsoid} with $r=r^*$ is an $(\varepsilon,\delta)$-solution. By Proposition~\ref{prop:box} applied in $\R^{2n}$, the balance of the economy and the required level of GDP are guaranteed if a planner keeps the outputs and final demands in the following intervals:
\begin{equation}\label{eq:IOintervals}
\begin{aligned}
&\bar x_j\Bigl(1+\zeta_j^0-\frac{r^*}{\sqrt{2n}}\Bigr)\le x_j\le\bar x_j\Bigl(1+\zeta_j^0+\frac{r^*}{\sqrt{2n}}\Bigr),\\
&\bar y_j\Bigl(1+\zeta_{n+j}^0-\frac{r^*}{\sqrt{2n}}\Bigr)\le y_j\le\bar y_j\Bigl(1+\zeta_{n+j}^0+\frac{r^*}{\sqrt{2n}}\Bigr),\qquad j=\overline{1,n}.
\end{aligned}
\end{equation}

Problem \eqref{eq:IOSP_obj}--\eqref{eq:IOSP_r} shows the structure of the robust plan explicitly. Adding constraints \eqref{eq:IOSP_sur} and \eqref{eq:IOSP_sho}, we obtain $\rho_i r\le\delta$, and constraints \eqref{eq:IOSP_corr} give $r\le\kappa$. The following statement gives a closed-form bound of the radius.

\begin{proposition}\label{prop:rho}
Let $H_i=\sum_{j=1}^nb_{ij}^2+s_i^2$, $i=\overline{1,n}$. Then $\rho_i^2=1-2b_{ii}+H_i$, and the following assertions hold:
\begin{enumerate}
\item[(a)] the balance constraints \eqref{eq:IOSP_sur}--\eqref{eq:IOSP_sho} of sector $i$ define a slab with the half-width $\delta/\rho_i$;
\item[(b)] $r^*(\varepsilon,\delta)\le r_0(\delta)=\min\{\kappa,\ \delta/\rho_{\max}\}$ for all $\varepsilon\ge0$, where $\rho_{\max}=\max_i\rho_i$;
\item[(c)] if $r_0(\delta)\le\theta$, then $r^*(\varepsilon,\delta)=r_0(\delta)$ for all $\varepsilon\ge\varepsilon_{\rm obs}(\delta)=f^*-\bar V+r_0(\delta)\norm{v\circ\bar x}$, and the observed table $\zeta=0$ is a center of a sphere with the maximum radius.
\end{enumerate}
\end{proposition}
\begin{proof}
The equality $\rho_i^2=1-2b_{ii}+H_i$ follows from \eqref{eq:rho}.

(a) The set $\{\zeta\mid|\zeta_i-\sum_{j=1}^nb_{ij}\zeta_j-s_i\zeta_{n+i}|\le\delta\}$ is a slab which is centered at $\zeta=0$, and the norm of its normal vector is $\rho_i$. Hence its half-width is $\delta/\rho_i$.

(b) Adding constraints \eqref{eq:IOSP_sur} and \eqref{eq:IOSP_sho}, we obtain $2\rho_ir\le2\delta$, $i=\overline{1,n}$, and constraints \eqref{eq:IOSP_corr} give $2r\le2\kappa$.

(c) Consider the point $\zeta^0=0$ and $r=r_0(\delta)$. Constraints \eqref{eq:IOSP_sur}--\eqref{eq:IOSP_corr} hold, since $\rho_ir_0\le\delta$ and $r_0\le\kappa$. Since $\mu_j\ge0$ and $\sum_{j=1}^n\mu_j=1$, we have $\norm{\mu}\le1$ and $r_0\norm{\mu}\le\theta$; the same holds for $\omega$, so constraints \eqref{eq:IOSP_res} hold. Constraint \eqref{eq:IOSP_gdp} takes the form $-r_0\norm{\sigma}\ge\eta$, which is equivalent to $\varepsilon\ge f^*-\bar V+r_0\bar V\norm{\sigma}=\varepsilon_{\rm obs}(\delta)$, since $\bar V\norm{\sigma}=\norm{v\circ\bar x}$. Hence $(0,r_0)$ is feasible for problem \eqref{eq:IOSP_obj}--\eqref{eq:IOSP_r}, $r^*\ge r_0$, and by (b) we obtain $r^*=r_0$.
\end{proof}

Since $b_{ij}\ge0$, $s_i>0$ and $\sum_{j=1}^nb_{ij}+s_i=1$, the number $H_i$ is the Herfindahl index of the distribution of sales of sector $i$ over its intermediate and final users, and $0<H_i\le1$. Therefore, $\rho_i^2\le2(1-b_{ii})$ and $\rho_i\le\sqrt2$ for every sector. For a given share $b_{ii}$ of the own use, $\rho_i$ increases with the concentration of sales, while a large own-use share decreases $\rho_i$: since the shares $b_{ij}$, $j\ne i$, and $s_i$ sum to $1-b_{ii}$, formula \eqref{eq:rho} gives $1-b_{ii}\le\rho_i\le\sqrt2\,(1-b_{ii})$, so that $\rho_i<1$ if $b_{ii}>1-1/\sqrt2\approx0.29$. If a sector does not use its own output ($b_{ii}=0$), then $\rho_i^2=1+H_i$ and $1<\rho_i\le\sqrt2$. If sector $i$ sells a share $u$ of its output to a single user other than itself, i.e., $u=b_{ij}$ for some $j\ne i$ or $u=s_i$, then $b_{ii}\le1-u$, and \eqref{eq:rho} gives $\rho_i^2\ge(1-b_{ii})^2+u^2\ge2u^2$, i.e., $\sqrt2\,u\le\rho_i\le\sqrt2$. It means that the balance of a sector which sells almost all its output to one other sector or to final demand is among the most sensitive to independent deviations of outputs and final demands, and the most sensitive balance determines the largest attainable margin.

Assertions (b) and (c) can also be explained by Theorem~\ref{thm:sym}. The balance conditions and the corridors, which give constraints \eqref{eq:IOSP_sur}--\eqref{eq:IOSP_corr}, describe the set $\{\zeta\mid|\zeta_i-\sum_{j=1}^nb_{ij}\zeta_j-s_i\zeta_{n+i}|\le\delta,\ i=\overline{1,n},\ |\zeta_k|\le\kappa,\ k=\overline{1,2n}\}$, which is convex, compact and centrally symmetric about the observed table $\zeta=0$. By Theorem~\ref{thm:sym}, $\zeta=0$ is its Chebyshev center, and since the distances from $\zeta=0$ to its bounding hyperplanes are $\delta/\rho_i$ and $\kappa$, the radius of the largest sphere inscribed into this set is $r_0(\delta)$. Assertion (c) means that the GDP and resource constraints \eqref{eq:IOSP_gdp} and \eqref{eq:IOSP_res} do not reduce this radius if $\varepsilon\ge\varepsilon_{\rm obs}(\delta)$ and $r_0(\delta)\le\theta$.

Let $r_0(\delta)\le\theta$. Then the point $\zeta^0=0$ satisfies constraints \eqref{eq:IOSP_sur}--\eqref{eq:IOSP_res} with $r=r_0(\delta)$ (see the proof of (c)), and these constraints define a nonempty compact set, since they are linear and constraints \eqref{eq:IOSP_corr} bound every $\zeta_j^0$. Denote by $\psi(\delta)\ge0$ the maximum of $\sum_{j=1}^n\sigma_j\zeta_j^0$ over this set. By (b), $r^*(\varepsilon,\delta)=r_0(\delta)$ if and only if some point of this set satisfies constraint \eqref{eq:IOSP_gdp} with $r=r_0(\delta)$, i.e., $\psi(\delta)-r_0(\delta)\norm{\sigma}\ge\eta$. By \eqref{eq:eta}, this is equivalent to $\varepsilon\ge\bar\varepsilon(\delta)$, where
\[
\bar\varepsilon(\delta)=f^*-\bar V\bigl(1+\psi(\delta)-r_0(\delta)\norm{\sigma}\bigr)\le\varepsilon_{\rm obs}(\delta).
\]
Thus, $r^*(\varepsilon,\delta)=r_0(\delta)$ if and only if $\varepsilon\ge\bar\varepsilon(\delta)$, i.e., the function $r^*(\cdot,\delta)$ reaches the plateau $r_0(\delta)$ at $\max\{0,\bar\varepsilon(\delta)\}$, and the threshold $\bar\varepsilon(\delta)$ is found by one additional linear programming problem. Moreover, if $r^*(\varepsilon,\delta)=\delta/\rho_i$ for some $i$, then the weights $\lambda=1/2$ on constraints \eqref{eq:IOSP_sur} and \eqref{eq:IOSP_sho} of sector $i$ and $\lambda=0$ on all other constraints form a solution of the dual problem \eqref{eq:dual}. Indeed, the normalized normals of these two constraints are opposite, and the dual value of these weights is $\delta/\rho_i=r^*$.

\subsection{Algorithm}

The proposed approach is summarized in Algorithm~\ref{alg:esp}.

\begin{algorithm}[H]
\caption{The $(\varepsilon,\delta)$-sphere packing algorithm for input--output analysis}
\label{alg:esp}
\begin{algorithmic}
\State \textbf{Input.} The observed input--output table (intermediate deliveries, gross outputs, final demands, value added, imported inputs and compensation of employees by sectors), the corridor half-width $\kappa$, the resource headroom $\theta$ and the tolerances $\varepsilon$ and $\delta$.
\State \textbf{Step 1.} Compute $A$, $v$, $m$ and $w$ from the input--output table and check that the column sums of $A$ are less than one.
\State \textbf{Step 2.} Solve the classical problem \eqref{eq:IO_obj}--\eqref{eq:IO_res} and find $f^*$, the optimal plan $\xi^*$ and the shadow prices.
\State \textbf{Step 3.} For given $\varepsilon$ and $\delta$, compute the coefficients \eqref{eq:rho}--\eqref{eq:eta} of problem \eqref{eq:IOSP_obj}--\eqref{eq:IOSP_r}.
\State \textbf{Step 4.} Solve problem \eqref{eq:IOSP_obj}--\eqref{eq:IOSP_r} and find the radius $r^*$ and the criticality weights $\lambda^*$.
\State \textbf{Step 5.} Among all centers of balls with the radius $r^*$, find the centers with the maximum GDP, i.e., maximize $\sum_{j=1}^n\sigma_j\zeta_j^0$ subject to constraints \eqref{eq:IOSP_gdp}--\eqref{eq:IOSP_res} with $r=r^*$. Among these centers, choose the center $\zeta^0$ closest to the observed table, i.e., minimize $\norm{\zeta^0}^2$.
\State \textbf{Step 6.} Compute the robust plan $x_j^0=\bar x_j(1+\zeta_j^0)$, $y_j^0=\bar y_j(1+\zeta_{n+j}^0)$, $j=\overline{1,n}$, the intervals \eqref{eq:IOintervals} and the largest cube.
\State \textbf{Step 7.} Check the robustness of $\xi^0$ and $\xi^*$ by Monte Carlo simulation.
\end{algorithmic}
\end{algorithm}

Steps 2 and 4 and the first part of Step 5 of Algorithm~\ref{alg:esp} require linear programming problems with at most $2n+1$ variables and at most $6n+3$ linear constraints besides $r\ge0$, and the second part of Step 5 is a strictly convex quadratic programming problem with the same constraints and GDP fixed at its maximum. If the ball with the maximum radius is not unique, Step~5 chooses its center with the maximum GDP, and if there are several such centers, the one closest to the observed table. These centers form a nonempty closed convex set, and the point of such a set closest to $\zeta=0$ is unique, so the robust plan is well defined. The parameters $\kappa$, $\theta$, $\varepsilon$ and $\delta$ are not estimated by the algorithm; they describe the planning scenario, and their choice is discussed in Section~\ref{sec:data}. The computations have been carried out in R (version \rVersion) with the packages \texttt{readxl}, \texttt{highs} (version \highsVersion, an interface to the HiGHS solver), and \texttt{quadprog} (version \quadprogVersion, for the quadratic problem of Step~5).

The main difference between the classical LP--IO analysis and the proposed sphere packing approach is the following. In the classical approach, a single optimal plan is found by problem \eqref{eq:IO_obj}--\eqref{eq:IO_res}, and the observed economy is compared with this plan \cite{tenRaaMohnen2002,LuptacikBohm2010}. By Proposition~\ref{prop:fragile}, such a plan lies on the boundary of the set of acceptable plans, and an arbitrarily small deviation of the plan in a certain direction violates a constraint. Unlike the classical approach, in the proposed sphere packing approach, we define the set $D_{IO}(\varepsilon,\delta)$ of acceptable plans and find a sphere with the maximum radius inscribed into this set by solving problem \eqref{eq:IOSP_obj}--\eqref{eq:IOSP_r}. Thus, the center of the sphere gives a robust plan, and its radius gives the admissible simultaneous deviations of all outputs and final demands. On the other hand, if $\varepsilon=0$ and $\delta=0$, then $D_{IO}(\varepsilon,\delta)$ is the set of solutions of the classical problem, and the radius is zero by Theorem~\ref{thm:pos}. In this sense, the classical LP--IO analysis is a limiting case of the proposed approach.

\section{Numerical experiments}\label{sec:numerical}

\subsection{Data and parameters}\label{sec:data}

In order to illustrate the proposed sphere packing approach numerically, we use the official symmetric input--output tables of Mongolia. The tables are of the non-competitive import type, are given in million tugriks (MNT) for 2010--2019 and 2024, and consist of 20 sectors. The residual sector ``other activities'' has practically no intermediate transactions and a negligible output (3.9 bn MNT in 2024, all of it delivered to final demand). Therefore, we merge it into ``other service activities'' and obtain $n=\nSectors$ sectors and $2n=\NDim$ variables. Note that in the English version of the table the seventeenth sector is labeled ``activities of households as employers'', whereas the Mongolian version shows that it is human health and social work; we use the latter label. Final demand is computed as $\bar y=\bar x-Z\mathbf 1$, which differs from the sum of final consumption, gross capital formation and exports by at most \maxBalErr\ million MNT. The column sums of $A$ do not exceed \colSumMax\ in all years, so the Leontief inverse exists and is nonnegative (Section~\ref{sec:iomodel}); the spectral radius of $A$ does not exceed \specRadiusMax\ (\specRadius\ in 2024).

In the numerical experiments, we use the following baseline parameters: $\kappa=\kappaPct\%$, $\theta=\thetaPct\%$, the balance tolerance $\delta=\deltaPct\%$ of the output of each sector on both sides of the balance, which can be interpreted as the admissible change in inventories, and the GDP tolerance $\varepsilon=\epsPct\%$ of $f^*$. In the sensitivity analysis, $\kappa$ varies from 5\% to 20\%, $\theta$ from 2\% to 10\%, $\delta$ from 0.5\% to 5\% and $\varepsilon$ from 0 to 5\%. These parameters are not estimated from the data. The published tables do not separate changes in inventories from gross capital formation, so the balance tolerance $\delta$ cannot be calibrated directly on observed inventories. Similarly, the corridors and the resource headroom describe the scenario of a planning agency rather than technological limits. For this reason, all four parameters are treated as scenario parameters, and the effect of each of them on the results is reported in Sections~\ref{sec:tradeoff} and~\ref{sec:sensitivity}.

\subsection{Structure of the Mongolian economy}\label{sec:structure}

Table~\ref{tab:structure} presents the structure of the Mongolian economy in 2024. The gross output was \outputTrn\ trillion MNT, final demand \finalTrn\ trillion MNT and GDP \gdpObsTrn\ trillion MNT. Mining produced \miningShareOutput\% of gross output and exported \miningExportShare\% of its output; its final-demand share $s_i$ and sales-concentration index $\rho_i$ are among the highest in the economy. The index $\rho_i$ is maximal for \rhoMaxSector\ ($\rho=\rhoMax$), which sells \rhoMaxShare\% of its output to final demand, followed by \rhoSecond, \rhoThird\ and \rhoFourth. \rhoMinSector\ has the smallest index ($\rho=\rhoMin$), since \selfUseMin\% of its output is used within the sector itself and the rest is distributed over many industries.

\begin{table}[H]
\caption{Structure of the Mongolian economy in 2024.}
\label{tab:structure}
\centering\small

\begin{tabular}{lrrrrrrr}
\toprule
Sector & $\bar x_i$ (trn MNT) & $s_i$ & $v_i$ & $m_i$ & $w_i$ & $H_i$ & $\rho_i$ \\
\midrule
Agriculture & 8.58 & 0.420 & 0.505 & 0.140 & 0.063 & 0.373 & 1.058 \\
Mining & 47.88 & 0.940 & 0.459 & 0.204 & 0.077 & 0.886 & 1.344 \\
Manufacturing & 17.60 & 0.691 & 0.330 & 0.129 & 0.083 & 0.491 & 1.155 \\
Electricity & 5.64 & 0.195 & 0.210 & 0.196 & 0.123 & 0.176 & 0.742 \\
Water supply & 0.79 & 0.277 & 0.531 & 0.131 & 0.388 & 0.121 & 1.002 \\
Construction & 15.28 & 0.604 & 0.212 & 0.321 & 0.096 & 0.392 & 1.060 \\
Trade & 20.51 & 0.437 & 0.483 & 0.194 & 0.164 & 0.236 & 1.078 \\
Transport & 13.35 & 0.318 & 0.346 & 0.265 & 0.145 & 0.200 & 0.922 \\
Accommodation & 3.15 & 0.754 & 0.355 & 0.125 & 0.171 & 0.589 & 1.255 \\
ICT & 3.88 & 0.400 & 0.424 & 0.210 & 0.176 & 0.192 & 1.005 \\
Finance & 6.78 & 0.151 & 0.632 & 0.130 & 0.184 & 0.261 & 1.055 \\
Real estate & 6.17 & 0.782 & 0.730 & 0.033 & 0.027 & 0.619 & 1.270 \\
Professional & 3.33 & 0.350 & 0.398 & 0.214 & 0.232 & 0.164 & 1.008 \\
Administrative & 2.71 & 0.349 & 0.255 & 0.240 & 0.147 & 0.165 & 1.024 \\
Public administration & 7.04 & 0.932 & 0.567 & 0.135 & 0.416 & 0.873 & 1.318 \\
Education & 4.93 & 0.969 & 0.738 & 0.046 & 0.643 & 0.939 & 1.393 \\
Health & 3.38 & 0.954 & 0.574 & 0.182 & 0.460 & 0.911 & 1.381 \\
Arts & 0.71 & 0.963 & 0.619 & 0.100 & 0.427 & 0.928 & 1.388 \\
Other services & 1.54 & 0.849 & 0.409 & 0.181 & 0.169 & 0.725 & 1.312 \\
\midrule Total / economy & 173.26 & 0.652 & 0.437 & 0.189 & 0.147 &  &  \\
\bottomrule
\end{tabular}
\\[1mm]
\parbox[b]{0.95\linewidth}{\footnotesize Notes: $\bar x_i$ is gross output; $s_i=\bar y_i/\bar x_i$ is the final-demand share; $v_i$, $m_i$ and $w_i$ are the value added, imported input and compensation coefficients; $H_i$ is the Herfindahl index of sales and $\rho_i$ is the sales-concentration index \eqref{eq:rho}. Source: NSO table DT\_NSO\_3100\_003V9; authors' calculations in R.}
\end{table}

\subsection{Solution of the classical problem}\label{sec:classical}

Problem \eqref{eq:IO_obj}--\eqref{eq:IO_res} has been solved numerically in R. The maximal GDP is $f^*=\fstarTrn$ trillion MNT, while the observed GDP is \gdpObsTrn\ trillion MNT. It means that the observed economy attains \effRatio\ of its LP potential, i.e., the maximal GDP exceeds the observed one by \gdpGap\%. At the optimal solution, both resource constraints are active: an additional MNT of imported inputs increases GDP by \dualImports\ MNT, and an additional MNT of wage bill by \dualWages\ MNT. In the optimal plan, \nYup\ final demands are at their upper bounds and \nYlow\ at their lower bounds, while no gross output reaches its corridor bound; the closest to its bound is the output of \xNearSector\ with the deviation \xNearPct\% (Table~\ref{tab:plans}). Final demand is raised to its upper bound in \yUpList, and lowered to its lower bound in \yLowList.

According to Proposition~\ref{prop:fragile}, the classical optimal solution has zero robustness: the largest ball centered at it has zero radius (the computed value is \rAtLP). On the other hand, the observed table is not an $(\varepsilon,\delta)$-solution for the baseline tolerances, since its GDP is \epsObs\% smaller than $f^*$.

\begin{table}[H]
\caption{Classical optimal plan $(x^*,y^*)$ and robust $(\varepsilon,\delta)$-plan $(x^0,y^0)$ in 2024, in \% deviation from the observed table ($\varepsilon=\epsPct\%$, $\delta=\deltaPct\%$).}
\label{tab:plans}
\centering\small

\begin{tabular}{lrrrr}
\toprule
Sector & $x^*_i$ & $y^*_i$ & $x^0_i$ & $y^0_i$ \\
\midrule
Agriculture & 9.03 & 10.00$^{\uparrow}$ & 8.88 & 9.33 \\
Mining & 9.87 & 10.00$^{\uparrow}$ & 9.30 & 9.33 \\
Manufacturing & 8.34 & 10.00$^{\uparrow}$ & 8.02 & 9.33 \\
Electricity & 2.11 & -10.00$^{\downarrow}$ & 1.33 & -9.33 \\
Water supply & 1.25 & -10.00$^{\downarrow}$ & 0.92 & -9.33 \\
Construction & -4.72 & -10.00$^{\downarrow}$ & -4.69 & -9.33 \\
Trade & 7.20 & 10.00$^{\uparrow}$ & 7.04 & 9.33 \\
Transport & 0.72 & -9.23 & 0.89 & -7.13 \\
Accommodation & 9.47 & 10.00$^{\uparrow}$ & 9.02 & 9.33 \\
ICT & 7.26 & 10.00$^{\uparrow}$ & 7.20 & 9.33 \\
Finance & 8.11 & 10.00$^{\uparrow}$ & 7.97 & 9.33 \\
Real estate & 9.26 & 10.00$^{\uparrow}$ & 8.82 & 9.33 \\
Professional & -0.47 & -10.00$^{\downarrow}$ & -0.69 & -9.33 \\
Administrative & 0.43 & -10.00$^{\downarrow}$ & 0.19 & -9.33 \\
Public administration & 10.00 & 10.00$^{\uparrow}$ & 9.33 & 9.21 \\
Education & 3.63 & 3.53 & 3.81 & 3.66 \\
Health & -9.14 & -10.00$^{\downarrow}$ & -8.59 & -9.33 \\
Arts & 9.69 & 10.00$^{\uparrow}$ & 9.10 & 9.33 \\
Other services & 9.06 & 10.00$^{\uparrow}$ & 8.59 & 9.33 \\
\midrule GDP (\% of $f^*$) & 100.00 &  & 99.73 &  \\
\bottomrule
\end{tabular}
\\[1mm]
\parbox[b]{0.95\linewidth}{\footnotesize Notes: $^{\uparrow}$ and $^{\downarrow}$ mark variables at their upper and lower corridor bounds ($\pm\kappaPct\%$) in the classical optimal plan. Every plan within the Euclidean distance $r^*=\rBase\%$ of the robust plan (in relative deviations) is an $(\varepsilon,\delta)$-solution.}
\end{table}

\subsection{Robust \texorpdfstring{$(\varepsilon,\delta)$}{(epsilon,delta)}-plan}\label{sec:robust}

Problem \eqref{eq:IOSP_obj}--\eqref{eq:IOSP_r} has been solved numerically in R. For the baseline parameters, $\eta=\etaBase\%$, i.e., every plan of the sphere must have GDP at least \etaBase\% larger than the observed GDP. The optimal radius is $r^*=\rBase\%$ in the \NDim-dimensional space of relative deviations. It means that every plan whose vector of relative deviations from the robust plan has the Euclidean norm at most \rBase\% is an $(\varepsilon,\delta)$-solution. Since $r^*/\sqrt{2n}=\boxBase\%$, intervals \eqref{eq:IOintervals} take the form
\begin{equation}\label{eq:robustsys}
(1+\zeta_j^0-\boxBase\%)\,\bar x_j\le x_j\le(1+\zeta_j^0+\boxBase\%)\,\bar x_j,\qquad (1+\zeta_{n+j}^0-\boxBase\%)\,\bar y_j\le y_j\le(1+\zeta_{n+j}^0+\boxBase\%)\,\bar y_j,
\end{equation}
$j=\overline{1,n}$, where the relative deviations $100\zeta_j^0$ and $100\zeta_{n+j}^0$ of the robust plan are given in the columns $x_j^0$ and $y_j^0$ of Table~\ref{tab:plans}. It means that every output and every final demand may deviate from the robust plan independently within $\pm\boxBase\%$ of its observed level. The half-width of the largest inscribed cube is \rCube\%. The robust plan keeps \gdpCentre\% of the maximal GDP, and its relative distance is \distCentreLP\ from the classical optimal plan and \distCentreObs\ from the observed table. Note that the center with the maximum GDP is not unique: the final demands of \centreFreeList\ can be changed within ranges of at most \centreFreeMax\ percentage points without changing the radius and GDP, and Step~5 of Algorithm~\ref{alg:esp} chooses the center closest to the observed table. All other components of the robust plan are the same for every center with the maximum GDP.

Table~\ref{tab:plans} and Figure~\ref{fig:plans} compare the robust plan with the classical optimal plan. Final demands which are at the corridor bounds in the classical solution are moved inside approximately by $r^*$, and outputs are adjusted so that every balance keeps a margin. The largest output adjustments, between \adjMin\ and \adjMax\ percentage points, occur in \adjTopList. Thus, the robust plan gives up less than 0.3\% of GDP in exchange for a guaranteed margin.

\begin{figure}[H]
\centering

\begin{tikzpicture}
\begin{groupplot}[paperfig, group style={group size=2 by 1, horizontal sep=0.5cm, yticklabels at=edge left},
  width=0.31\textwidth, height=0.46\textwidth, xmin=-11, xmax=11, xtick={-10,-5,0,5,10},
  ymin=0.4, ymax=19.6, ytick={1,...,19}, yticklabels={{Arts},{Water supply},{Other services},{Administrative},{Accommodation},{Professional},{Health},{ICT},{Education},{Electricity},{Real estate},{Finance},{Public administration},{Agriculture},{Transport},{Construction},{Manufacturing},{Trade},{Mining}}]
\nextgroupplot[title={Gross output}, legend to name=figThreeLegend, legend columns=2,
  legend style={draw=none, font=\scriptsize, text=figInkTwo,
    /tikz/every even column/.append style={column sep=12pt}}]
\draw[figBase, line width=0.4pt] (-10,0.4) -- (-10,19.6) (10,0.4) -- (10,19.6);
\draw[figInkTwo, line width=0.4pt] (0,0.4) -- (0,19.6);
\addplot[only marks, figOrange, figdot] coordinates {(9.688339,1.000000) (1.247756,2.000000) (9.059496,3.000000) (0.433445,4.000000) (9.466792,5.000000) (-0.465992,6.000000)
  (-9.138522,7.000000) (7.261243,8.000000) (3.632244,9.000000) (2.105918,10.000000) (9.257391,11.000000) (8.106621,12.000000)
  (9.997404,13.000000) (9.031344,14.000000) (0.722141,15.000000) (-4.718990,16.000000) (8.344628,17.000000) (7.201875,18.000000)
  (9.870231,19.000000)};
\addlegendentry{classical LP optimum}
\addplot[only marks, figBlue, figdot] coordinates {(9.104361,1.000000) (0.915931,2.000000) (8.591535,3.000000) (0.186164,4.000000) (9.019319,5.000000) (-0.688422,6.000000)
  (-8.590890,7.000000) (7.199042,8.000000) (3.810087,9.000000) (1.327360,10.000000) (8.819162,11.000000) (7.968700,12.000000)
  (9.327427,13.000000) (8.875927,14.000000) (0.888569,15.000000) (-4.692084,16.000000) (8.024336,17.000000) (7.038362,18.000000)
  (9.300366,19.000000)};
\addlegendentry{robust $(\varepsilon,\delta)$-plan}
\nextgroupplot[title={Final demand}]
\draw[figBase, line width=0.4pt] (-10,0.4) -- (-10,19.6) (10,0.4) -- (10,19.6);
\draw[figInkTwo, line width=0.4pt] (0,0.4) -- (0,19.6);
\addplot[only marks, figOrange, figdot] coordinates {(10.000000,1.000000) (-10.000000,2.000000) (10.000000,3.000000) (-10.000000,4.000000) (10.000000,5.000000) (-10.000000,6.000000)
  (-10.000000,7.000000) (10.000000,8.000000) (3.528093,9.000000) (-10.000000,10.000000) (10.000000,11.000000) (10.000000,12.000000)
  (10.000000,13.000000) (10.000000,14.000000) (-9.234096,15.000000) (-10.000000,16.000000) (10.000000,17.000000) (10.000000,18.000000)
  (10.000000,19.000000)};
\addplot[only marks, figBlue, figdot] coordinates {(9.327427,1.000000) (-9.327348,2.000000) (9.327427,3.000000) (-9.327427,4.000000) (9.327427,5.000000) (-9.327427,6.000000)
  (-9.327427,7.000000) (9.327427,8.000000) (3.659365,9.000000) (-9.327427,10.000000) (9.327427,11.000000) (9.327427,12.000000)
  (9.207657,13.000000) (9.327427,14.000000) (-7.134583,15.000000) (-9.327427,16.000000) (9.327427,17.000000) (9.327427,18.000000)
  (9.327427,19.000000)};
\end{groupplot}
\node[anchor=south] at ($(group c1r1.north west)!0.5!(group c2r1.north east)+(0,16pt)$)
  {\pgfplotslegendfromname{figThreeLegend}};
\node[anchor=north, font=\footnotesize, text=figInkTwo]
  at ($(group c1r1.south west)!0.5!(group c2r1.south east)+(0,-14pt)$)
  {deviation from the observed 2024 table (\%)};
\end{tikzpicture}
\caption{Classical optimal plan and robust $(\varepsilon,\delta)$-plan by sectors, in \% deviation from the observed 2024 table. Vertical lines show the corridor bounds $\pm\kappaPct\%$; sectors are ordered by gross output.}
\label{fig:plans}
\end{figure}
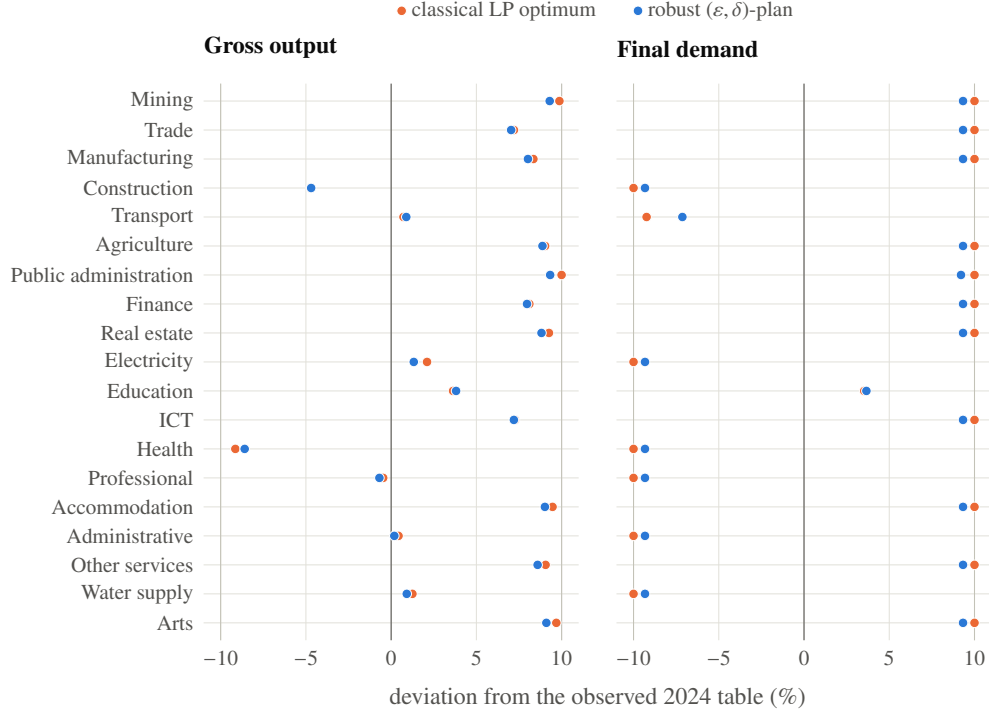

\subsection{Robustness--optimality trade-off}\label{sec:tradeoff}

Table~\ref{tab:grid} and Figure~\ref{fig:tradeoff} show the optimal radius $r^*$ as a function of $\varepsilon$ and $\delta$. In accordance with Theorem~\ref{thm:concave}, every curve is non--decreasing, concave and piecewise linear. We distinguish the following two cases.

\emph{Case 1} (optimality-constrained regime). For small $\varepsilon$, the radius increases almost linearly with the slope of about \slopeEps\ at $\delta=\deltaPct\%$, since the ball must stay above the GDP constraint near the optimal vertex. Even for $\varepsilon=0$, the radius is positive (\rEpsZero\% at $\delta=\deltaPct\%$), because the balance tolerance admits plans with GDP larger than $f^*$ (see Remark~\ref{rem:eps0}).

\emph{Case 2} (balance-constrained regime). For $\varepsilon\ge\bar\varepsilon(\delta)$, the radius equals the plateau $\delta/\rho_{\max}$ of Proposition~\ref{prop:rho}, which is determined by the balance of \rhoMaxSector. The threshold is almost proportional to the balance tolerance, $\bar\varepsilon(\delta)\approx\kinkRatio\,\delta$; for instance, $\bar\varepsilon=\epsKinkBase\%$ for $\delta=\deltaPct\%$.

The sufficient condition of Proposition~\ref{prop:rho}(c) holds for $\varepsilon\ge\epsObsPlateau\%$, and then the observed table itself is a center of a ball with the maximum radius. However, the plateau is reached much earlier, near the optimum. It means that a GDP tolerance of about 0.6 of the balance tolerance is sufficient to obtain the largest margin which the balance tolerance allows, whereas a GDP tolerance of one half of the balance tolerance is not (Table~\ref{tab:grid}).

\begin{table}[H]
\caption{Optimal radius $r^*(\varepsilon,\delta)$ in \% of relative deviation (2024).}
\label{tab:grid}
\centering\small

\begin{tabular}{lrrrr}
\toprule
$\varepsilon$ (\% of $f^*$) & $\delta=0.5\%$ & $\delta=1.0\%$ & $\delta=2.0\%$ & $\delta=5.0\%$ \\
\midrule
0.00\% & 0.131 & 0.235 & 0.376 & 0.539 \\
0.10\% & 0.217 & 0.328 & 0.475 & 0.647 \\
0.25\% & 0.336 & 0.461 & 0.621 & 0.807 \\
0.50\% & 0.359 & 0.673 & 0.863 & 1.073 \\
1.00\% & 0.359 & 0.718 & 1.333 & 1.606 \\
2.00\% & 0.359 & 0.718 & 1.436 & 2.672 \\
5.00\% & 0.359 & 0.718 & 1.436 & 3.590 \\
\midrule $\delta/\rho_{\max}$ & 0.359 & 0.718 & 1.436 & 3.590 \\
$\bar\varepsilon(\delta)$ (\% of $f^*$) & 0.28 & 0.56 & 1.12 & 2.89 \\
\bottomrule
\end{tabular}
\\[1mm]
\parbox[b]{0.9\linewidth}{\footnotesize Notes: $\delta$ is the balance tolerance in \% of sectoral output and $\varepsilon$ is the GDP tolerance in \% of $f^*$. The row $\delta/\rho_{\max}$ gives the plateau of Proposition~\ref{prop:rho}, and $\bar\varepsilon(\delta)$ is the smallest $\varepsilon$ at which the plateau is reached.}
\end{table}

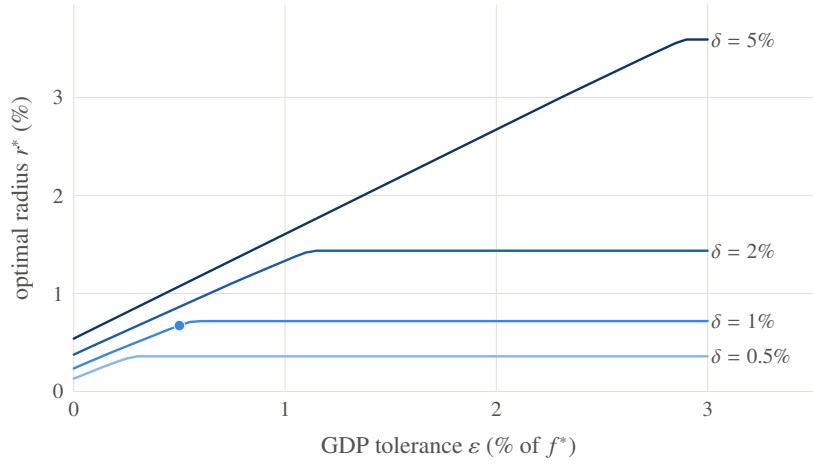
\begin{figure}[H]
\centering

\def\figTwoLabA{0.5}
\def\figTwoLabB{1}
\def\figTwoLabC{2}
\def\figTwoLabD{5}
\def\figTwoBaseX{0.500000}
\def\figTwoBaseY{0.672573}

\begin{tikzpicture}
\begin{axis}[paperfig, width=0.62\textwidth, height=0.32\textwidth, xmin=0, xmax=3.55, ymin=0,
  xtick={0,1,2,3}, xlabel={GDP tolerance $\varepsilon$ (\% of $f^*$)}, ylabel={optimal radius $r^*$ (\%)}]
\addplot[rampA, figline] coordinates {(0.000000,0.131257) (0.050000,0.174212) (0.100000,0.217162) (0.150000,0.260111) (0.200000,0.299636) (0.250000,0.336287)
  (0.300000,0.359034) (0.350000,0.359034) (0.400000,0.359034) (0.450000,0.359034) (0.500000,0.359034) (0.550000,0.359034)
  (0.600000,0.359034) (0.650000,0.359034) (0.700000,0.359034) (0.750000,0.359034) (0.800000,0.359034) (0.850000,0.359034)
  (0.900000,0.359034) (0.950000,0.359034) (1.000000,0.359034) (1.050000,0.359034) (1.100000,0.359034) (1.150000,0.359034)
  (1.200000,0.359034) (1.250000,0.359034) (1.300000,0.359034) (1.350000,0.359034) (1.400000,0.359034) (1.450000,0.359034)
  (1.500000,0.359034) (1.550000,0.359034) (1.600000,0.359034) (1.650000,0.359034) (1.700000,0.359034) (1.750000,0.359034)
  (1.800000,0.359034) (1.850000,0.359034) (1.900000,0.359034) (1.950000,0.359034) (2.000000,0.359034) (2.050000,0.359034)
  (2.100000,0.359034) (2.150000,0.359034) (2.200000,0.359034) (2.250000,0.359034) (2.300000,0.359034) (2.350000,0.359034)
  (2.400000,0.359034) (2.450000,0.359034) (2.500000,0.359034) (2.550000,0.359034) (2.600000,0.359034) (2.650000,0.359034)
  (2.700000,0.359034) (2.750000,0.359034) (2.800000,0.359034) (2.850000,0.359034) (2.900000,0.359034) (2.950000,0.359034)
  (3.000000,0.359034)} node[pos=1, anchor=west, figlabel] {$\delta=\figTwoLabA\%$};
\addplot[rampB, figline] coordinates {(0.000000,0.235233) (0.050000,0.282074) (0.100000,0.327642) (0.150000,0.373211) (0.200000,0.418102) (0.250000,0.461139)
  (0.300000,0.504128) (0.350000,0.547080) (0.400000,0.590030) (0.450000,0.632979) (0.500000,0.672573) (0.550000,0.709225)
  (0.600000,0.718068) (0.650000,0.718068) (0.700000,0.718068) (0.750000,0.718068) (0.800000,0.718068) (0.850000,0.718068)
  (0.900000,0.718068) (0.950000,0.718068) (1.000000,0.718068) (1.050000,0.718068) (1.100000,0.718068) (1.150000,0.718068)
  (1.200000,0.718068) (1.250000,0.718068) (1.300000,0.718068) (1.350000,0.718068) (1.400000,0.718068) (1.450000,0.718068)
  (1.500000,0.718068) (1.550000,0.718068) (1.600000,0.718068) (1.650000,0.718068) (1.700000,0.718068) (1.750000,0.718068)
  (1.800000,0.718068) (1.850000,0.718068) (1.900000,0.718068) (1.950000,0.718068) (2.000000,0.718068) (2.050000,0.718068)
  (2.100000,0.718068) (2.150000,0.718068) (2.200000,0.718068) (2.250000,0.718068) (2.300000,0.718068) (2.350000,0.718068)
  (2.400000,0.718068) (2.450000,0.718068) (2.500000,0.718068) (2.550000,0.718068) (2.600000,0.718068) (2.650000,0.718068)
  (2.700000,0.718068) (2.750000,0.718068) (2.800000,0.718068) (2.850000,0.718068) (2.900000,0.718068) (2.950000,0.718068)
  (3.000000,0.718068)} node[pos=1, anchor=west, figlabel] {$\delta=\figTwoLabB\%$};
\addplot[rampC, figline] coordinates {(0.000000,0.375862) (0.050000,0.425697) (0.100000,0.474943) (0.150000,0.523478) (0.200000,0.572013) (0.250000,0.620548)
  (0.300000,0.669083) (0.350000,0.717618) (0.400000,0.766153) (0.450000,0.814688) (0.500000,0.863223) (0.550000,0.911631)
  (0.600000,0.959911) (0.650000,1.008191) (0.700000,1.056472) (0.750000,1.104752) (0.800000,1.151402) (0.850000,1.196970)
  (0.900000,1.242539) (0.950000,1.287963) (1.000000,1.333224) (1.050000,1.378486) (1.100000,1.418449) (1.150000,1.436135)
  (1.200000,1.436135) (1.250000,1.436135) (1.300000,1.436135) (1.350000,1.436135) (1.400000,1.436135) (1.450000,1.436135)
  (1.500000,1.436135) (1.550000,1.436135) (1.600000,1.436135) (1.650000,1.436135) (1.700000,1.436135) (1.750000,1.436135)
  (1.800000,1.436135) (1.850000,1.436135) (1.900000,1.436135) (1.950000,1.436135) (2.000000,1.436135) (2.050000,1.436135)
  (2.100000,1.436135) (2.150000,1.436135) (2.200000,1.436135) (2.250000,1.436135) (2.300000,1.436135) (2.350000,1.436135)
  (2.400000,1.436135) (2.450000,1.436135) (2.500000,1.436135) (2.550000,1.436135) (2.600000,1.436135) (2.650000,1.436135)
  (2.700000,1.436135) (2.750000,1.436135) (2.800000,1.436135) (2.850000,1.436135) (2.900000,1.436135) (2.950000,1.436135)
  (3.000000,1.436135)} node[pos=1, anchor=west, figlabel] {$\delta=\figTwoLabC\%$};
\addplot[rampD, figline] coordinates {(0.000000,0.539369) (0.050000,0.593415) (0.100000,0.646720) (0.150000,0.700025) (0.200000,0.753330) (0.250000,0.806635)
  (0.300000,0.859941) (0.350000,0.913246) (0.400000,0.966551) (0.450000,1.019856) (0.500000,1.073161) (0.550000,1.126466)
  (0.600000,1.179771) (0.650000,1.233076) (0.700000,1.286381) (0.750000,1.339686) (0.800000,1.392991) (0.850000,1.446296)
  (0.900000,1.499602) (0.950000,1.552907) (1.000000,1.606212) (1.050000,1.659517) (1.100000,1.712822) (1.150000,1.766127)
  (1.200000,1.819432) (1.250000,1.872737) (1.300000,1.926042) (1.350000,1.979347) (1.400000,2.032652) (1.450000,2.085957)
  (1.500000,2.139263) (1.550000,2.192568) (1.600000,2.245873) (1.650000,2.299178) (1.700000,2.352483) (1.750000,2.405788)
  (1.800000,2.459093) (1.850000,2.512398) (1.900000,2.565703) (1.950000,2.619008) (2.000000,2.672313) (2.050000,2.725618)
  (2.100000,2.778923) (2.150000,2.832229) (2.200000,2.885534) (2.250000,2.938839) (2.300000,2.992144) (2.350000,3.044296)
  (2.400000,3.095957) (2.450000,3.147619) (2.500000,3.199263) (2.550000,3.250759) (2.600000,3.302255) (2.650000,3.353750)
  (2.700000,3.405246) (2.750000,3.455785) (2.800000,3.505557) (2.850000,3.555222) (2.900000,3.590339) (2.950000,3.590339)
  (3.000000,3.590339)} node[pos=1, anchor=west, figlabel] {$\delta=\figTwoLabD\%$};
\addplot[only marks, rampB, figdot, mark size=2pt] coordinates {(\figTwoBaseX,\figTwoBaseY)};
\end{axis}
\end{tikzpicture}
\caption{The optimal radius $r^*$ as a function of the GDP tolerance $\varepsilon$ for four balance tolerances $\delta$ (2024). The dot shows the baseline.}
\label{fig:tradeoff}
\end{figure}

\subsection{Critical constraints}\label{sec:critical}

Table~\ref{tab:duals} presents the criticality weights of Proposition~\ref{prop:dual} for three values of $\varepsilon$. In these three cases, and in all cases of Table~\ref{tab:sensitivity}, the computed solution of problem \eqref{eq:IOSP_obj}--\eqref{eq:IOSP_r} is a nondegenerate vertex with exactly $2n+1$ active constraints. Hence the dual solution is unique, and by Proposition~\ref{prop:dual} the weight $\lambda_k^*=\partial r^*/\partial\hat q_k$ is the rate at which the radius increases when the $k$-th constraint is moved outwards. For the baseline tolerances, the largest weight corresponds to the GDP constraint (\dualGDP), followed by the \dualTopTwoName\ (\dualTopTwo), the final-demand cap of mining, imports, wages and the balances of manufacturing, trade and agriculture. Near the optimum, the GDP requirement increases mining output, while mining exports are bounded by the final-demand corridor. Therefore, the margin of the plan depends on the amount of unsold mining output which can be absorbed, for example, as stockpiles. For $\varepsilon=0$, this role is played by the capacity of mining. In the balance-constrained regime ($\varepsilon=2\%$), both sides of the balance of education have the equal weights 0.5, which is the dual solution described after Proposition~\ref{prop:rho}.

Note that the robustness-critical sectors differ from the key sectors of the classical linkage analysis. According to the criterion of Rasmussen \cite{Rasmussen1956}, where both the backward and forward linkages exceed the average, the key sectors in 2024 are \rasmussenKey. None of them appears among the constraints with the largest weights in Table~\ref{tab:duals}, and every constraint of these sectors has a weight smaller than 0.005. Linkages measure how strongly shocks propagate, whereas the criticality weights show which constraints limit the margin of a near-optimal plan.

\begin{table}[H]
\caption{Criticality weights $\lambda^*$ of the active constraints (2024, $\delta=\deltaPct\%$).}
\label{tab:duals}
\centering\small

\begin{tabular}{lrrr}
\toprule
Constraint & $\varepsilon=0$ & $\varepsilon=0.5\%$ & $\varepsilon=2\%$ \\
\midrule
GDP floor & 0.326 & 0.248 & -- \\
Balance, surplus side: Education & -- & -- & 0.500 \\
Balance, shortage side: Education & -- & -- & 0.500 \\
Imports & 0.120 & 0.094 & -- \\
Wages & 0.101 & 0.078 & -- \\
Balance, surplus side: Mining & -- & 0.148 & -- \\
Output capacity: Mining & 0.139 & -- & -- \\
Final-demand cap: Mining & -- & 0.104 & -- \\
Balance, surplus side: Manufacturing & 0.041 & 0.047 & -- \\
Balance, surplus side: Trade & 0.042 & 0.033 & -- \\
Output capacity: Real estate & 0.048 & -- & -- \\
Balance, surplus side: Agriculture & -- & 0.031 & -- \\
\bottomrule
\end{tabular}
\\[1mm]
\parbox[b]{0.9\linewidth}{\footnotesize Notes: the weights are dual solutions of problem \eqref{eq:dual} and sum to one in each column (``--'' denotes zero); only constraints with one of the six largest weights in some column are shown. The GDP floor, the surplus and shortage sides of the balances, the output capacities and final-demand caps, and imports and wages correspond to constraints \eqref{eq:IOSP_gdp}, \eqref{eq:IOSP_sur}, \eqref{eq:IOSP_sho}, \eqref{eq:IOSP_corr} and \eqref{eq:IOSP_res}, respectively.}
\end{table}

\subsection{Sensitivity analysis}\label{sec:sensitivity}

Table~\ref{tab:sensitivity} shows the results for different corridor widths $\kappa$ and resource headroom $\theta$. The ratio GDP$/f^*$ varies from \effMin\ to \effMax, and the radius for $\varepsilon=\epsPct\%$ varies from \rSensMin\% to \rSensMax\%. For $\varepsilon=2\%$, the radius equals the plateau \rPlateau\% in all twelve cases. By Proposition~\ref{prop:rho}, the plateau $r_0(\delta)=\delta/\rho_{\max}$ depends only on $\delta$ and on the sales structure, since $\delta/\rho_{\max}<\kappa$ and $\delta/\rho_{\max}<\theta$ in all cases, whereas the fact that $\varepsilon=2\%$ exceeds the threshold $\bar\varepsilon(\delta)$ in every case is a result of the computation; the sufficient condition $\varepsilon\ge\varepsilon_{\rm obs}(\delta)$ of Proposition~\ref{prop:rho}(c) does not hold for $\varepsilon=2\%$.

The effect of resources is non-monotone. For $\kappa=\kappaPct\%$, an increase of $\theta$ from 5\% to 10\% raises $f^*$ by \fstarThetaTenChg\%, but decreases the baseline radius from \rThetaFive\% to \rThetaTen\%. The reason is that the optimum moves to a sharper vertex, where \activeTen\ constraints are active instead of \activeFive, including all \xBoundsTen\ output corridors.

We also examined three additional cases. If shortages are not allowed, i.e., the balance condition is replaced by $0\le S_i-D_i\le\delta\bar x_i$, then the width of every balance slab is halved. Adding the two sides of the balance of sector $i$ as in the proof of Proposition~\ref{prop:rho}(b), we obtain $2\rho_ir\le\delta$, so that the radius cannot exceed $\delta/(2\rho_{\max})=\rOneSidedTheory\%$. The computed radius is \rOneSidedBase\% for both $\varepsilon=\epsPct\%$ and $\varepsilon=2\%$, i.e., this bound is attained. If the ball is measured in monetary units instead of relative deviations, its radius is \rAbsBn\ billion MNT, and the implied simultaneous tolerance varies from $\pm\absMinPct\%$ of the base value (\absMinSector) to $\pm\absMaxPct\%$ (\absMaxSector). Such a ball is determined by the scale of the variables rather than by the structure of the economy, which confirms the necessity of the relative coordinates. Finally, the largest inscribed cube has the half-width \rCube\% per variable, while the Euclidean ball guarantees $\pm\boxBase\%$.

\begin{table}[H]
\caption{Sensitivity to the corridor width $\kappa$ and the resource headroom $\theta$ (2024, $\delta=\deltaPct\%$).}
\label{tab:sensitivity}
\centering\small

\begin{tabular}{rrrrrrl}
\toprule
$\kappa$ & $\theta$ & $f^*$ (trn MNT) & GDP$/f^*$ & $r^*$, $\varepsilon=0.5\%$ & $r^*$, $\varepsilon=2\%$ & Second critical constraint \\
\midrule
5\% & 2\% & 77.93 & 0.971 & 0.665 & 0.718 & Imports \\
5\% & 5\% & 79.48 & 0.952 & 0.386 & 0.718 & Output capacity: Mining \\
5\% & 10\% & 79.48 & 0.952 & 0.386 & 0.718 & Output capacity: Mining \\
10\% & 2\% & 78.66 & 0.962 & 0.672 & 0.718 & Imports \\
10\% & 5\% & 80.81 & 0.937 & 0.673 & 0.718 & Balance, surplus side: Mining \\
10\% & 10\% & 83.27 & 0.909 & 0.404 & 0.718 & Output capacity: Mining \\
15\% & 2\% & 79.35 & 0.954 & 0.675 & 0.718 & Imports \\
15\% & 5\% & 81.66 & 0.927 & 0.686 & 0.718 & Imports \\
15\% & 10\% & 84.72 & 0.893 & 0.689 & 0.718 & Balance, surplus side: Mining \\
20\% & 2\% & 80.04 & 0.946 & 0.679 & 0.718 & Imports \\
20\% & 5\% & 82.41 & 0.919 & 0.691 & 0.718 & Imports \\
20\% & 10\% & 85.92 & 0.881 & 0.696 & 0.718 & Balance, surplus side: Mining \\
\bottomrule
\end{tabular}
\\[1mm]
\parbox[b]{0.95\linewidth}{\footnotesize Notes: $r^*$ is given in \% of relative deviation. The last column gives the constraint with the second largest criticality weight at $\varepsilon=\epsPct\%$; the largest weight corresponds to the GDP constraint in every case.}
\end{table}

\subsection{Monte Carlo simulation}\label{sec:mc}

In order to examine the behavior of plans under random deviations, we perturb every output and final demand independently, $\xi_j\mapsto\xi_j+\bar\xi_ju_j$, where $u_j$ is uniformly distributed on $[-s,s]$. For each size $s$, \mcNsim\ perturbed plans have been generated, and the plans which remain in $D_{IO}(\varepsilon,\delta)$ have been counted (Figure~\ref{fig:mc}). Not a single perturbed classical optimal plan remains admissible, even for $s=0.01\%$. For the center of the sphere, all perturbed plans remain admissible up to $s=\mcFullSphere\%$, which is almost three times larger than the guaranteed tolerance \boxBase\%. For $s=0.4\%$ and $s=0.5\%$, the shares are \mcSphereFour\% and \mcSphereFive\%, and a half of the plans remain admissible at $s\approx\mcHalfSphere\%$. The center of the largest cube, chosen as in Step~5, is completely robust up to its half-width \rCube\%, but the share decreases to \mcCubeFour\% at $s=0.4\%$ and to \mcCubeFive\% at $s=0.5\%$. The cube is designed for the worst corner of a box, whereas the Euclidean ball keeps a larger distance to the numerous constraints involving a single variable (the corridors) and to the balance slabs. Therefore, the sphere is more robust against typical deviations caused by independent shocks.

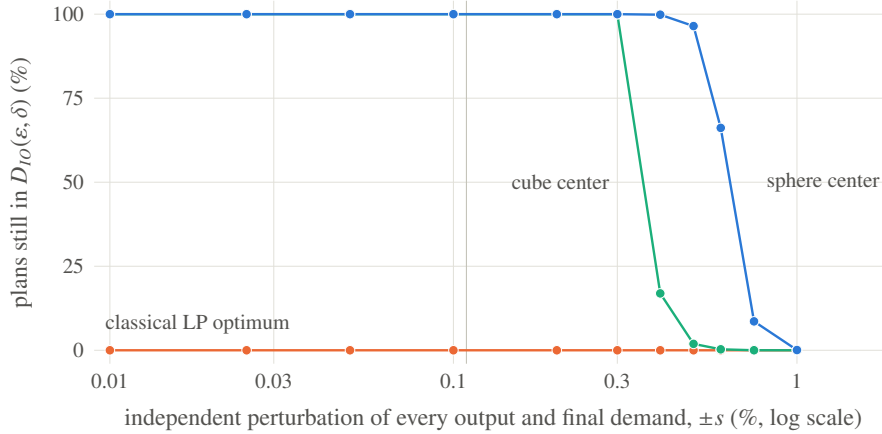
\begin{figure}[H]
\centering

\def\figFourBox{0.109106}

\begin{tikzpicture}
\begin{axis}[paperfig, width=0.66\textwidth, height=0.30\textwidth, xmode=log, xmin=0.009, xmax=1.9,
  xtick={0.01,0.03,0.1,0.3,1}, xticklabels={0.01,0.03,0.1,0.3,1}, ymin=-4, ymax=104,
  ytick={0,25,50,75,100},
  xlabel={independent perturbation of every output and final demand, $\pm s$ (\%, log scale)},
  ylabel={plans still in $D_{IO}(\varepsilon,\delta)$ (\%)}]
\draw[figBase, line width=0.4pt] (\figFourBox,-4) -- (\figFourBox,104);
\addplot[figOrange, figline, figdot] coordinates {(0.010000,0.000000) (0.025000,0.000000) (0.050000,0.000000) (0.100000,0.000000) (0.200000,0.000000) (0.300000,0.000000)
  (0.400000,0.000000) (0.500000,0.000000) (0.600000,0.000000) (0.750000,0.000000) (1.000000,0.000000)};
\addplot[figGreen, figline, figdot] coordinates {(0.010000,100.000000) (0.025000,100.000000) (0.050000,100.000000) (0.100000,100.000000) (0.200000,100.000000) (0.300000,100.000000)
  (0.400000,16.920000) (0.500000,1.905000) (0.600000,0.285000) (0.750000,0.010000) (1.000000,0.000000)};
\addplot[figBlue, figline, figdot] coordinates {(0.010000,100.000000) (0.025000,100.000000) (0.050000,100.000000) (0.100000,100.000000) (0.200000,100.000000) (0.300000,100.000000)
  (0.400000,99.845000) (0.500000,96.455000) (0.600000,66.165000) (0.750000,8.600000) (1.000000,0.075000)};
\node[anchor=south west, figlabel] at (0.0095,4) {classical LP optimum};
\node[anchor=east, figlabel] at (0.29,50) {cube center};
\node[anchor=west, figlabel] at (0.8,50) {sphere center};
\end{axis}
\end{tikzpicture}
\caption{Share of randomly perturbed plans which remain $(\varepsilon,\delta)$-solutions (2024, $\varepsilon=\epsPct\%$, $\delta=\deltaPct\%$, \mcNsim\ plans for each size). The thin vertical line shows the guaranteed tolerance $r^*/\sqrt{2n}$ of the sphere.}
\label{fig:mc}
\end{figure}

\subsection{Results for 2010--2024}\label{sec:panel}

Table~\ref{tab:panel} presents the results for all eleven tables. The observed efficiency is stable, between \effPanelMin\ and \effPanelMax, and so is the baseline radius, between \rPanelMin\% and \rPanelMax\%. The plateau is determined by education in all years except 2010, when health had the most concentrated sales. At the same time, the structure of mining changed considerably: its export share increased from \miningExpFirst\% in 2010 to \miningExpLast\% in 2024, and its sales-concentration index changed from \rhoMiningFirst\ to \rhoMiningLast\ (Figure~\ref{fig:rho}). Mining now has one of the narrowest balance slabs, and in the optimality-constrained regime its balance is already among the critical constraints (Table~\ref{tab:duals}). This is a geometric counterpart of the concentration risk discussed in the literature on the Dutch disease in Mongolia \cite{Dagys2020,Dagys2025}.

\begin{table}[H]
\caption{Classical optimal solution and robustness, 2010--2024 ($\varepsilon=\epsPct\%$, $\delta=\deltaPct\%$, $\kappa=\kappaPct\%$, $\theta=\thetaPct\%$).}
\label{tab:panel}
\centering\small

\begin{tabular}{rrrrrrlr}
\toprule
Year & GDP (trn) & $f^*$ (trn) & GDP$/f^*$ & $r^*$ (\%) & $\delta/\rho_{\max}$ (\%) & Sector with $\rho_{\max}$ & Mining exp.\ (\%) \\
\midrule
2010 & 8.76 & 9.37 & 0.934 & 0.682 & 0.733 & Health (1.363) & 69 \\
2011 & 11.52 & 12.31 & 0.936 & 0.676 & 0.716 & Education (1.397) & 77 \\
2012 & 14.83 & 15.83 & 0.937 & 0.701 & 0.715 & Education (1.399) & 70 \\
2013 & 17.00 & 18.19 & 0.934 & 0.698 & 0.716 & Education (1.396) & 74 \\
2014 & 20.15 & 21.57 & 0.934 & 0.699 & 0.720 & Education (1.389) & 89 \\
2015 & 21.01 & 22.49 & 0.934 & 0.691 & 0.720 & Education (1.388) & 90 \\
2016 & 21.94 & 23.48 & 0.934 & 0.685 & 0.719 & Education (1.390) & 96 \\
2017 & 25.47 & 27.31 & 0.933 & 0.664 & 0.722 & Education (1.385) & 92 \\
2018 & 29.06 & 31.18 & 0.932 & 0.661 & 0.718 & Education (1.392) & 84 \\
2019 & 33.83 & 36.26 & 0.933 & 0.670 & 0.718 & Education (1.393) & 89 \\
2024 & 75.70 & 80.81 & 0.937 & 0.673 & 0.718 & Education (1.393) & 96 \\
\bottomrule
\end{tabular}
\\[1mm]
\parbox[b]{0.95\linewidth}{\footnotesize Notes: GDP and $f^*$ are given in trillion MNT at current prices. The last column is the share of mining output which is exported. The tables for 2020--2023 are not available.}
\end{table}

\begin{figure}[H]
\centering

\def\figFiveValA{1.392626}
\def\figFiveLabA{1.395714}
\def\figFiveNameA{Education}
\def\figFiveValB{1.388049}
\def\figFiveLabB{1.388214}
\def\figFiveNameB{Arts}
\def\figFiveValC{1.380714}
\def\figFiveLabC{1.380714}
\def\figFiveNameC{Health}
\def\figFiveValD{1.344015}
\def\figFiveLabD{1.344015}
\def\figFiveNameD{Mining}
\def\figFiveValE{1.317565}
\def\figFiveLabE{1.317565}
\def\figFiveNameE{Public administration}

\begin{tikzpicture}
\begin{axis}[paperfig, width=0.6\textwidth, height=0.29\textwidth, xmin=2009.5, xmax=2024.5,
  xtick={2010,2012,...,2024}, xticklabel style={/pgf/number format/1000 sep={}},
  yticklabel style={/pgf/number format/fixed, /pgf/number format/precision=3,
    /pgf/number format/fixed zerofill}, ylabel={sales-concentration index $\rho$}]
\addplot[figBlue, figline, line width=0.8pt, figdot, mark size=1.6pt] coordinates {(2010.000000,1.363133) (2011.000000,1.397411) (2012.000000,1.398823) (2013.000000,1.396143) (2014.000000,1.388985) (2015.000000,1.387941)
  (2016.000000,1.390071) (2017.000000,1.384876) (2018.000000,1.391813) (2019.000000,1.393200)};
\addplot[figOrange, figline, line width=0.8pt, figdot, mark size=1.6pt] coordinates {(2010.000000,1.295206) (2011.000000,1.278199) (2012.000000,1.344744) (2013.000000,1.279141) (2014.000000,1.277925) (2015.000000,1.333129)
  (2016.000000,1.328204) (2017.000000,1.321193) (2018.000000,1.313866) (2019.000000,1.316780)};
\addplot[figGreen, figline, line width=0.8pt, figdot, mark size=1.6pt] coordinates {(2010.000000,1.363373) (2011.000000,1.309784) (2012.000000,1.358100) (2013.000000,1.365330) (2014.000000,1.355916) (2015.000000,1.375596)
  (2016.000000,1.378559) (2017.000000,1.379362) (2018.000000,1.375000) (2019.000000,1.378909)};
\addplot[figYellow, figline, line width=0.8pt, figdot, mark size=1.6pt] coordinates {(2010.000000,1.325526) (2011.000000,1.369808) (2012.000000,1.358200) (2013.000000,1.356926) (2014.000000,1.364485) (2015.000000,1.278286)
  (2016.000000,1.299782) (2017.000000,1.320273) (2018.000000,1.291061) (2019.000000,1.299589)};
\addplot[figPink, figline, line width=0.8pt, figdot, mark size=1.6pt] coordinates {(2010.000000,1.294043) (2011.000000,1.330608) (2012.000000,1.351072) (2013.000000,1.349940) (2014.000000,1.360962) (2015.000000,1.378630)
  (2016.000000,1.383107) (2017.000000,1.376128) (2018.000000,1.368996) (2019.000000,1.378865)};
\addplot[only marks, figBlue, figdot, mark size=2.1pt] coordinates {(2024,\figFiveValA)};
\addplot[only marks, figOrange, figdot, mark size=2.1pt] coordinates {(2024,\figFiveValB)};
\addplot[only marks, figGreen, figdot, mark size=2.1pt] coordinates {(2024,\figFiveValC)};
\addplot[only marks, figYellow, figdot, mark size=2.1pt] coordinates {(2024,\figFiveValD)};
\addplot[only marks, figPink, figdot, mark size=2.1pt] coordinates {(2024,\figFiveValE)};
\draw[figMuted, line width=0.3pt] (2024.3,\figFiveValA) -- (2024.8,\figFiveLabA);
\draw[figMuted, line width=0.3pt] (2024.3,\figFiveValB) -- (2024.8,\figFiveLabB);
\draw[figMuted, line width=0.3pt] (2024.3,\figFiveValC) -- (2024.8,\figFiveLabC);
\draw[figMuted, line width=0.3pt] (2024.3,\figFiveValD) -- (2024.8,\figFiveLabD);
\draw[figMuted, line width=0.3pt] (2024.3,\figFiveValE) -- (2024.8,\figFiveLabE);
\node[anchor=west, figlabel] at (2024.9,\figFiveLabA) {\figFiveNameA};
\node[anchor=west, figlabel] at (2024.9,\figFiveLabB) {\figFiveNameB};
\node[anchor=west, figlabel] at (2024.9,\figFiveLabC) {\figFiveNameC};
\node[anchor=west, figlabel] at (2024.9,\figFiveLabD) {\figFiveNameD};
\node[anchor=west, figlabel] at (2024.9,\figFiveLabE) {\figFiveNameE};
\end{axis}
\end{tikzpicture}
\caption{Sales-concentration index $\rho_i$ of the five sectors with the most concentrated sales in 2024. The observations for 2019 and 2024 are not connected, since the tables for 2020--2023 are not available.}
\label{fig:rho}
\end{figure}
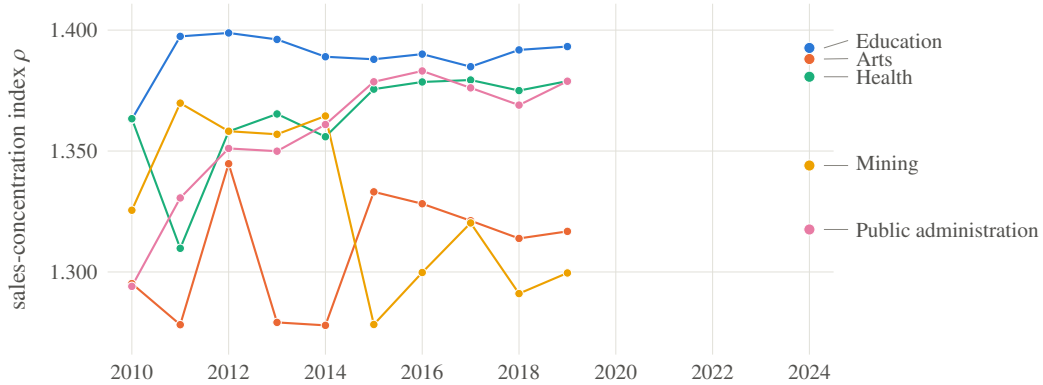

\section{Discussion}\label{sec:discussion}

The results of this paper address the three research gaps identified in Section~\ref{sec:gap}. Regarding (G1), the robust $(\varepsilon,\delta)$-plan is accompanied by tolerance ranges which are guaranteed for simultaneous deviations of all outputs and final demands ($\pm\boxBase\%$ in 2024), whereas the classical optimal plan has no margin at all. Regarding (G2), the $\delta$-relaxation makes the sphere packing approach applicable to balance equations, and Theorems~\ref{thm:lp}--\ref{thm:concave} and Propositions~\ref{prop:slab}--\ref{prop:dual} describe the optimal radius as a function of the tolerances. Regarding (G3), Proposition~\ref{prop:rho} relates the largest attainable margin to the sales concentration of sectors, and the computations show that the constraints which limit the margin differ from the key sectors of the linkage analysis. Below, we compare these results with the related literature.

The classical LP--IO models \cite{DSS1958,CheneryKretschmer1956,Manne1974} find an optimal vertex and shadow prices, and efficiency studies compare the observed economy with this frontier \cite{tenRaaMohnen2002,LuptacikBohm2010}. Our classical solution confirms this view: the observed economy of Mongolia attains about 93--94\% of its LP potential in 2010--2019 and 2024. The sphere packing approach gives an additional object, namely, a robust plan in the interior of the set of near-optimal plans together with a guaranteed margin, and the Monte Carlo simulation shows that the difference between the two plans is important in practice. While the LP models in \cite{PichlerFarmer2022} examine the boundary of the feasible set of an IO economy under shocks, the proposed approach describes its interior.

The tolerance approach \cite{Wendell1985,Filippi2005,Borgonovo2018} and robust optimization \cite{Soyster1973,BenTalNemirovski2000,BertsimasSim2004,Kang2020EJOR} protect the optimal basis or the constraints against perturbations of the data. In contrast, the $(\varepsilon,\delta)$-sphere protects the decision itself and shows how far the plan may deviate. These two views complement each other. Moreover, the criticality weights are dual variables of a linear programming problem, and since the dual solutions are unique in our computations (Section~\ref{sec:critical}), they have the usual sensitivity interpretation (Proposition~\ref{prop:dual}).

In \cite{EnkhbatTungalag2023,Chuluunbaatar2023,Ganlkhagva2023,Enkhbat2021NACO}, spheres have been inscribed into sets defined by inequalities of a single decision maker, such as the profitability conditions of a company, financial ratios, a budget surplus and a regression model of a mineral process. The balance equations of an interdependent economy require the $\delta$-relaxation (Proposition~\ref{prop:empty} and Theorem~\ref{thm:pos}), while the $\varepsilon$-relaxation connects the approach with the near-optimal spaces \cite{Brill1982,NeumannBrown2021}. Our results also show that the sphere is more robust against random independent deviations than the cube of the survival theory \cite{Ashchepkov1980,Badam2021}, although the cube gives a larger worst-case tolerance per variable.

From the economic point of view, the balance tolerance, that is, the capacity to hold inventories, makes robust planning possible, and it increases the robustness at the rate $1/\rho_{\max}$. A GDP tolerance of about 0.6 of the balance tolerance is sufficient to use it completely. The largest margin is limited by sectors with concentrated sales: public services, which sell almost exclusively to final consumption, and, increasingly, the export-oriented mining sector. Therefore, a diversification of the users of mining output, for example, by domestic processing, would not only reduce the exposure of the economy \cite{Dagys2025} but also widen the balance slab which restricts robust plans. Near the optimum, the robustness is limited by imports and labor, whose shadow prices \dualImports\ and \dualWages\ are considerable.

In practice, a planning agency can present the robust plan together with its tolerance ranges ($\pm\boxBase\%$ guaranteed, and $\pm\mcFullSphere\%$ without a single violation in the simulation) instead of a single optimal plan without any margin. Then sectoral ministries can set their targets within these ranges, and the criticality weights show which constraints should be relaxed first.

The proposed approach has some limitations. The analysis is based on a static Leontief model with fixed coefficients and on tables in current prices. The corridors, the resource headroom and the tolerances are scenario parameters; therefore, we reported the sensitivity of the results instead of claiming particular values. The Euclidean ball treats all directions of deviations equally, and the guaranteed interval $\pm r^*/\sqrt{2n}$ decreases with the number of variables; for this reason, we also reported the cube and the Monte Carlo results. Finally, the analysis considers quantities and does not model prices.

\section{Conclusions}\label{sec:conclusion}

We have introduced the notion of $(\varepsilon,\delta)$-solutions of a linear programming problem with equality constraints and extended input--output analysis with linear programming by the sphere packing approach. Since the balance equations of the input--output model have no interior, the classical problem $\max\ip{c}{x}$, $Ax=b$, has been relaxed into the $(\varepsilon,\delta)$-set, in which GDP may be smaller than its optimal value by $\varepsilon$ and supply may differ from demand by $\delta$, and a sphere with the maximum radius has been inscribed into this set. We have shown that the obtained problem reduces to linear programming and that its optimal radius is positive for positive tolerances, non--decreasing, concave and piecewise linear with respect to the tolerances, and bounded by the narrowest balance constraint. In the relative coordinates, this bound equals $\delta/\rho_{\max}$, where $\rho_i$ is the sales-concentration index, and if it does not exceed the corridor half-width and the resource headroom, it is attained for all GDP tolerances not smaller than a threshold, which is found by one additional linear programming problem. Every classical optimal solution has zero robustness, and the criticality weights identify the constraints which limit the robustness.

The advantage of the sphere packing approach over the classical LP--IO analysis is that it finds a set of production plans which ensure the balance of the economy and a given level of GDP. Instead of a single optimal plan without any margin, a planning agency obtains a plan together with the ranges of outputs and final demands, given by the center and the radius of the sphere, within which both requirements are guaranteed.

For the Mongolian economy, the robust plan costs less than 0.3\% of the maximal GDP, while the classical optimal plan tolerates no deviation at all. Near the optimum, the robustness is limited by the GDP requirement, mining, imports and wages, and in the balance-constrained regime by the concentrated sales of public services. The efficiency, the radius and the plateau are stable over all available years, and the plateau is the same for all considered corridors and resource headrooms.

The results have been obtained for a static Leontief model with fixed coefficients and for tables in current prices. The tolerances, the corridors and the resource headroom are scenario parameters, and the Euclidean ball treats all directions of deviations equally. In further research, the tolerances and corridors can be estimated from observed data on inventories and capacity utilization, and the analysis can be repeated with constant-price tables. The ball can be replaced by ellipsoids calibrated on observed shocks. Moreover, employment, capital and environmental constraints can be added, and the approach can be applied to regional and multiregional input--output tables and to dynamic input--output models with investment.

\section*{Acknowledgments}
The authors thank the University of Finance and Economics, Ulaanbaatar, Mongolia.

\section*{Author contributions}
Munkhbayar Byambadash: Conceptualization, Writing--review and editing. Tumendelger Lkhagvasuren: Formal analysis, Data curation, Visualization, Writing--original draft. Enkhbat Rentsen: Methodology, Supervision.

\section*{Funding}
This research was supported by the University of Finance and Economics, Ulaanbaatar, Mongolia.

\section*{Data availability}
The input--output tables are publicly available from the database of the National Statistics Office of Mongolia. The R code that reproduces all tables, figures and numbers is available from the corresponding author on request.

\section*{Conflict of interest}
The authors declare no conflict of interest.

\section*{Use of Generative-AI tools declaration}
During the preparation of this manuscript, the authors used Claude (Anthropic) to check the numerical results, to write parts of the R code, and for typesetting and language editing. All AI-assisted content was critically reviewed and verified by the authors, who take full responsibility for the content of the paper.

\end{document}